\documentclass[journal,twocolumn]{IEEEtran}
\usepackage[normalem]{ulem}
\usepackage{blindtext, graphicx}
\usepackage{listings}
\usepackage{graphicx}
\usepackage[font=small]{caption}
\usepackage{color}
\usepackage{amsmath,bm}
\usepackage{amsmath}
\usepackage{amssymb}
\usepackage[ruled,vlined]{algorithm2e}
\usepackage{algpseudocode}
\usepackage{amsthm}
\usepackage{makecell}
\usepackage{siunitx}
\usepackage{comment}
\usepackage{cite}
\usepackage[colorlinks=true, citecolor=blue]{hyperref}
\usepackage{booktabs}
\usepackage{multirow}
\usepackage{mathtools}
\usepackage[utf8]{inputenc}
\makeatletter
\let\NAT@parse\undefined
\makeatother
\newtheorem{theorem}{Theorem}

\newtheorem{proposition}{Proposition}
\newtheorem{definition}{Definition}
\newtheorem{remark}{Remark}

\title{Improving Sequential Market Coordination via Value-oriented Renewable Energy Forecasting}

\author{Yufan Zhang, Honglin Wen~\IEEEmembership{Member,~IEEE,} Yuexin Bian, Yuanyuan Shi~\IEEEmembership{Member,~IEEE.}
\vspace{-3em}
\thanks{Yufan Zhang is with the College of Engineering, Cornell University and was with the Department of Electrical and Computer Engineering, University of California, San Diego. \\
\indent Honglin Wen is with the Department of Electrical Engineering, Shanghai Jiao Tong University, and Dyson School of Design Engineering, Imperial College London. \\
\indent Yuexin Bian and Yuanyuan Shi are with the Department of Electrical and Computer Engineering, University of California San Diego.}
}

\begin{document}

\maketitle
\thispagestyle{empty}
\pagestyle{plain}
%\emph{arabic}

%%%%%%%%%%%%%%%%%%%%%%%%%%%%%%%%%%%%%%%%%%%%%%%%%%%%%%%%%%%%%%%%%%%%%%%%%%%%%%%%
\begin{abstract}
Large penetration of renewable energy sources (RESs) brings huge uncertainty into the electricity markets. The current deterministic clearing approach in the day-ahead (DA) market, where RESs participate based on expected production, has been criticized for causing a lack of coordination between the DA and real-time (RT) markets, leading to high overall operating costs. Previous works indicate that improving day-ahead RES entering quantities can significantly mitigate the drawbacks of deterministic clearing. In this work, we propose using a trained forecasting model, referred to as \emph{value-oriented forecasting}, to determine RES Improved Entering Quantities (RIEQ) more efficiently during the operational phase. Unlike traditional models that minimize statistical forecasting errors, our approach trains model parameters to minimize the expected overall operating costs across both DA and RT markets. We derive the exact form of the loss function used for training, which becomes piecewise linear when market clearing is modeled by linear programs. Additionally, we provide the analytical gradient of the loss function with respect to the forecast, enabling an efficient training  strategy. Numerical studies demonstrate that our forecasts significantly reduce overall operating costs for deterministic market clearing compared to conventional forecasts based on expected RES production.

% Large penetration of renewable energy sources (RESs) brings huge uncertainty into the electricity markets. While existing deterministic market clearing fails to accommodate the uncertainty, the recently proposed stochastic market clearing struggles to achieve desirable market properties. In this work, we propose a value-oriented forecasting approach, which tactically determines the RESs generation that enters the day-ahead market. With such a forecast, the existing deterministic market clearing framework can be maintained, and the day-ahead and real-time overall operating cost is reduced. At the training phase, the forecasting model parameters are estimated to minimize expected day-ahead and real-time overall operating costs, instead of minimizing forecast errors in a statistical sense. Theoretically, we derive the exact form of the loss function for training the forecasting model that aligns with such a goal. For market clearing modeled by linear programs, this loss function is a piecewise linear function. Additionally, we derive the analytical gradient of the loss function with respect to the forecast, which inspires an efficient training strategy. A numerical study shows our forecasts can bring significant benefits of the overall cost reduction to deterministic market clearing, compared to quality-oriented forecasting approach.


\end{abstract}

%%%%%%%%%%%%%%%%%%%%%%%%%%%%%%%%%%%%%%%%%%%%%%%%%%%%%%%%%%%%%%%%%%%%%%%%%%%%%%%%
\section{Introduction}
% \subsection{\todo{Background and Motivation}}
The current short-term electricity markets are organized in a sequence of trading floors, i.e., day-ahead (DA) and real-time (RT) markets \cite{morales2013integrating}. A DA market is cleared 12-36 hours before the actual operation. A RT market runs close to the delivery time and addresses any imbalance from the DA schedules. They are initially designed for controllable fossil-fueled generators in the view of traditional power system operation. However, the increasing share of renewable energy sources (RESs) (up to 30\% of global electricity generation in 2022 \cite{RES}) exposes the electricity markets to significant uncertainty and therefore raises concerns to the market operation \cite{morales2013integrating}.

% \cite{dvorkin2019chance,kazempour2017value}. 

A significant challenge with sequential deterministic market clearing arises from its limited ability to coordinate the two markets under large RES uncertainty. Specifically, DA clearing decisions are made under RES uncertainty, influencing RT market clearing, which subsequently manages power imbalances arising from discrepancies between RES schedules and realizations. Currently, the DA market is cleared without accounting for its impact on RT market clearing. As a result, DA and RT markets operate without coordination, leading to high overall operating costs \cite{kazempour2017value}. 
% Without properly considering the impact of DA decisions,  which are made under RES uncertainty, on RT market clearing, DA and RT markets operate without coordination, resulting in high overall operating costs \cite{kazempour2017value} \todo{this sentence is confusing. break it into two and make it more clear}.
% The separation of DA and RT market clearings means that DA market clearing does not adequately consider the re-dispatch costs incurred due to RES uncertainty \cite{kazempour2017value}. This often results in higher overall operating costs. 
Consequently, stochastic market clearing has been proposed, which informs the DA market of the operating cost in the RT market \cite{morales2012pricing}. % \todo{This approach involves the joint clearing of the DA and RT markets, incorporating settlement for both the DA generation schedule and the future regulation power provided by flexible units at RT \cite{morales2012pricing}. Usually we do not say `joint clearing', since we only inform DA clearing with RT markets.}
% high market efficiency with minimized expected operating costs.
While stochastic market clearing can improve overall economic efficiency, it struggles to simultaneously maintain key market properties, such as revenue adequacy and cost recovery \cite{zavala2017stochastic}. To address this, efforts have been made to preserve these desirable properties within stochastic market clearing. For instance, \cite{kazempour2018stochastic} guarantees cost recovery and revenue adequacy both per scenario and in expectation, though this comes at the cost of reduced market efficiency.

Previous studies have shown that combining DA deterministic market clearing with RES improved entering quantities (RIEQ) enhances coordination between DA and RT markets, achieving performance comparable to stochastic clearing \cite{morales2014electricity,zhao2022uncertainty}. The key idea is to tactically determine RIEQ, which is made under uncertainty and used as input parameters in DA clearing, by considering the impact of DA clearing results on future RT market clearing. For instance, studies \cite{morales2014electricity,zhao2022uncertainty} optimize RIEQ by solving a bilevel stochastic program on a case-by-case basis during the operational phase. Similarly, \cite{dvorkin2018setting} applies these methods to predict reserve requirements, facilitating coordination of reserve and energy markets. While these approaches hold promise for improving market coordination, they introduce computational challenges in determining RIEQ.

Building on this foundation, we aim to explore an important question: \emph{{What if the RES entering quantities in the DA market were improved as RIEQ by a ``properly'' trained forecasting model, specifically optimized to minimize overall system operating costs?} If so, could the sequential markets be better coordinated?} Here, the overall system operating cost in both the DA and RT markets is used as the quantitative metric to evaluate market coordination.

With the trained forecasting model for determining RIEQ, the computational burden at the operational phase  \cite{morales2014electricity,zhao2022uncertainty} can be avoided. However, aligning the training objective of a forecasting model with the operational goal is a significant challenge, falling within the realm of value-oriented forecasting \cite{stratigakos2022prescriptive,chen2021feature,morales2023prescribing}. Several research threads have emerged to address this challenge, encompassing integrated optimization, differentiable programming, and the loss function design.

For the first thread of research on integrated optimization, forecasting model parameters are optimized concurrently with decision variables \cite{chen2021feature,garcia2021application}. In the same vein, \cite{morales2023prescribing} introduces a bilevel program where DA market clearing forms the lower level with the RIEQ prediction as a parameter, while the upper level optimizes both model parameters and RT decisions.  This method requires a linear forecasting model to make the program solvable by commercial optimization solvers, which may limit the performance.
The second thread of research, on differentiable programming, enables the use of more advanced forecasting models, such as neural networks, by deriving the gradient of decision solution w.r.t. the forecast \cite{donti2017task,wahdany2023more}. However, obtaining such a gradient involves repeatedly solving the inverse of a large-size matrix resulting from the Karush–Kuhn–Tucker conditions, making it computationally expensive. In a general sense, the first two threads implicitly design value-oriented loss functions by aligning forecasting model training with decision-making value, without explicitly formulating a loss function that links forecasts to the decision-making objective. The third approach, and the focus of our work, centers on explicitly designing value-oriented loss functions. Compared to the implicit loss function, the explicit loss function has a clearer structure, which enhances the explainability of value-oriented forecasts \cite{zhang2024valueoriented}. The existing ones, such as pinball loss \cite{pinson2007trading} or Smart ``Predict, then Optimize" (SPO) loss \cite{elmachtoub2022smart}, are primarily designed for single-stage stochastic programs  (e.g., the one without considering flexible unit redispatch in the RT market). Thus, how to develop a value-oriented loss function for sequential market clearing remains an open question.

To address the challenge of training the RIEQ forecasting model, we propose a tailored loss function that uses the RIEQ forecast as the input and the overall DA and RT system operating costs as the output. The forecasting model parameter estimation is formulated as a bilevel program: the upper level optimizes forecasting model parameters, while the lower level solves the DA and RT market clearing problems using the RIEQ forecast as the input parameter. To link operating cost with the RIEQ forecast, we leverage the lower-level \emph{dual} problems and replace the upper-level objective with dual objectives. The reformulated upper-level objective is then transformed into an analytical function of the forecast, serving as the loss function for training. This requires the derivation of functions that connect the forecast with primal and dual solutions. Specifically, the dual solutions are linked to forecasts through the dual problems of DA and RT market clearing, while the functions for primal solutions are derived from the active constraints of their respective primal problems. Substituting these functions into the upper-level objective yields a value-oriented loss function for training. Our main contributions are twofold: from the market perspective and from the methodological perspective of obtaining RIEQ,

  1) From the market perspective, the proposed approach maintains deterministic DA market clearing while improving sequential market coordination with RIEQ forecasts. Additionally, we theoretically demonstrate that deterministic DA market clearing with RIEQ forecasts maintains key market properties, namely cost recovery and revenue adequacy.

  2) From the methodological perspective of obtaining RIEQ, we propose determining RIEQ via a trained forecasting model at the operational phase. During the training phase, we analytically derive a value-oriented loss function that aligns the forecasting model’s training objective with the operational value, aiming to minimize overall system operating costs in both DA and RT markets. We reveal that the loss function is piecewise linear for market clearing problems modeled by linear programs, with each piece associated with active constraints. Leveraging insights from the loss function structure, we propose a computationally efficient training approach for gradient descent-based methods, which recalculates the gradient only when encountering new active constraints.

  % 2) From a theoretical perspective, we derive a value-oriented loss function that aligns the forecasting model's training objective with operational value, minimizing overall system operating costs in both DA and RT markets.
  
  % 3) From a practitioner's perspective, the analytical loss function allows the analytical derivation of the gradient. Leveraging the structure of the gradient, we propose a computationally efficient training approach.
  
  % which is more computationally efficient compared to differentiable programming \cite{donti2017task}. 

  The above contributions distinguish this work from our prior study \cite{zhang2024valueoriented}. This work directly addresses the coordination challenge posed by RESs in electricity markets. Additionally, the method for deriving the loss function is specifically tailored to the market clearing problem and is more generalized compared to the approaches in \cite{zhang2024valueoriented}. Unlike the previous work \cite{zhang2024valueoriented}, where the DA and RT operating problems are connected solely through the forecasts, in this work, the DA and RT market clearing are linked through the DA decisions, such as the schedules of traditional and renewable generators, which are implicitly influenced by the forecasts. This is specific to the market setting, as the schedules of flexible resources are determined in the DA market, while their flexibility enables them to make RT adjustments based on the DA schedules to correct RT power imbalances. The influence of the forecasts on the RT clearing is indirect, since RES forecasts are not explicitly included in the RT clearing process. The method for deriving the loss function in this work is specifically designed to account for this implicit impact, extending beyond the approach proposed in \cite{zhang2024valueoriented}.

Also, it is worth mentioning that this work has not fully addressed all the challenges of implementing RIEQ forecasts in practical market clearing, such as how to align RIEQ with RES producer offering quantities. Our primary goal is to demonstrate that RIEQ can be determined via trained forecasting models at the operational phase, making it computationally convenient. Additionally, the RIEQ can enhance the coordination between day-ahead and real-time markets, thereby reducing overall operating costs while preserving key market properties, such as cost recovery and revenue adequacy.

The remaining parts of this paper are organized as follows. The preliminaries regarding the sequential market clearing are given in Section \uppercase\expandafter{\romannumeral2}. Section \uppercase\expandafter{\romannumeral3} formulates a bilevel program for forecasting model parameter estimation. Section \uppercase\expandafter{\romannumeral4} derives the loss function for value-oriented forecasting and the training process is presented in Section \uppercase\expandafter{\romannumeral5}. Results are discussed and evaluated in Section \uppercase\expandafter{\romannumeral6}, followed by the conclusions.

\textbf{Notations:} 
% We let uppercase letters in boldface such as $\bm{X}$ represent matrices. 
The notation $\bm{X}[\mathcal{J}]$ signifies the sub-matrix comprised of rows from the matrix $\bm{X}$ whose indices are included in the index set $\mathcal{J}$. Considering the column vectors $\bm{x}_1$ and $\bm{x}_2$, the expression $\bm{x}=[\bm{x}_1;\bm{x}_2]$ indicates the vertical concatenation of $\bm{x}_1$ and $\bm{x}_2$, resulting in the formation of $\bm{x}$. The operator $\Pi_{\mathcal{J}}$ is used to extract a segment from the vector, such as retrieving $\bm{x}_1$ from $\bm{x}$.

\begin{table*}[ht]
\centering
\caption{Nomenclature}
\label{tab:nomenclature}
\renewcommand{\arraystretch}{1.15}
\setlength{\tabcolsep}{6pt} % Reduced from 10pt to 6pt

% Define narrower column widths
\newcolumntype{L}{>{\raggedright\arraybackslash}p{5cm}} % Reduced from 6cm
\newcolumntype{R}{>{\raggedright\arraybackslash}p{5cm}} % Reduced from 6cm

\begin{tabular}{ll|ll}
\toprule
\multicolumn{2}{l}{\textbf{Sets}} \\
\midrule
$\mathcal{J}^a_{DA,d}$ & \makecell[l]{Row index set of active constraints of day-ahead clearing\\ on day $d$} & $\mathcal{J}^a_{RT,d,\tau}$ & \makecell[l]{Row index set of active constraints of real-time clearing\\ at time $\tau$ on day $d$.}\\
\midrule
\multicolumn{2}{l}{\textbf{Variables}} \\
\midrule
$\bm{p}_{d,\tau}$ & Traditional generator day-ahead schedule at time $\tau$ on day $d$.
& $\bm{w}_{d,\tau}$ & Renewable generator day-ahead schedule at time $\tau$ on day $d$. \\
$\bm{p}_{d,\tau}^+$ & \makecell[l]{Traditional generator real-time adjustment for up-regulation\\ at time $\tau$ on day $d$.}
& $\bm{p}_{d,\tau}^-$ & \makecell[l]{Traditional generator real-time adjustment for down-regulation\\ at time $\tau$ on day $d$.}  \\
$\bm{\kappa}_{d,\tau}$ & Renewable generation spill at time $\tau$ on day $d$.
& $\bm{x}_{d}$ & Collection of day-ahead decisions on day $d$. \\
$\bm{z}_{d,\tau}$ & Collection of real-time decisions at time $\tau$ on day $d$.
& $\Theta$ & Forecasting model parameter.\\
$\bm{p}_{d,\tau}^{+-}$ & \makecell[l]{Collection of real-time up-/down- regulation adjustment\\ at time $\tau$ on day $d$.} & $\bm{\sigma}_{d}$ & day-ahead dual decisions on day $d$.\\
$\bm{\nu}_{d,\tau}$ & \makecell[l]{Real-time dual decision at time $\tau=1$ on day $d$.} & $\bm{\xi}_{d,\tau}$ & \makecell[l]{Real-time dual decision at time $\tau=2,...,T$ on day $d$.}\\
\midrule
\multicolumn{2}{l}{\textbf{Parameters}} \\
\midrule
$\bm{\rho}$ & Marginal day-ahead generation cost.
& $\bm{l}_{d,\tau}$ & Electricity load at time $\tau$ on day $d$. \\
$\overline{\bm{f}}$ & Transmission line capacity
& $\overline{\bm{p}}$ & Max day-ahead output of traditional generator  \\
$\overline{\bm{r}}$ & Max ramping capacity of traditional generator
& $\bm{\rho}_+$ & Marginal real-time generation cost for up-regulation. \\
$\bm{\rho}_-$ & Marginal real-time generation utility for down-regulation.
& $\bm{y}_{d,\tau}$ & Renewable generation realization \\
$\bm{H}$ & Power transfer distribution factors 
& $\bm{s}_{d,\tau}$ & Forecasting context \\
$\overline{\bm{y}}_{d,\tau}$ & Renewable energy capacity
  \\
\bottomrule
\end{tabular}
\end{table*}

\section{Preliminaries}
The framework and mathematical formulation of sequential market clearing are introduced in subsection \ref{Sequential Market Clearing}, and the reformulation is presented in subsection \ref{Mathematical Reformulation}.

% The bilevel program for forecasting model parameter estimation is described in subsection \ref{Parameter Estimation}.

% \begin{figure}[t]
%   \centering
%   % Requires \usepackage{graphicx}
%   \includegraphics[scale=0.6]{Figures/sequentialmarket.pdf}
%   \caption{An illustration of forecasting and the sequential market clearing where $\hat{\bm{y}}_{d,\tau}$ and $\bm{y}_{d,\tau}$ are forecasts and realization, and $\bm{w}_{d,\tau}^*$ is the RES schedule.}
% \label{sequential}
% \end{figure}
% \vspace{-1em}

\subsection{Sequential Market Clearing}\label{Sequential Market Clearing}

We consider the sequential clearing of DA and RT markets \cite{morales2013integrating}. The DA market is cleared at time $t$ on day $d-1$, with an advance of $k$ hours in time to the next day $d$, and covers energy transactions on day $d$, typically on an hourly basis. RES production is uncertain during the DA market, where the RES forecast acts as the RIEQ. Due to inevitable forecasting errors, the energy imbalance caused by forecasting errors needs to be settled in the RT market. Concretely, in the DA market, the operator determines the schedules of generators and RES to satisfy inelastic demand. The generation and RES schedule for each time-slot $\tau,\forall \tau=1,...,24$ on the next day $d$ are denoted as $\bm{p}_{d,\tau}$ and $\bm{w}_{d,\tau}$, respectively. The DA market clearing is,
\begin{subequations}\label{DA}
\begin{alignat}{2}
& \mathop{\min}_{\bm{x}_d}   &&\sum_{\tau=1}^T\bm{\rho}^\top\bm{p}_{d,\tau}\label{DAa}\\ 
    & \text{s.t.} &&  \bm{1}^\top(\bm{p}_{d,\tau}+\bm{w}_{d,\tau})=\bm{1}^\top\bm{l}_{d,\tau}:\gamma_{d,\tau},\forall \tau=1,...,T\label{DAb}
    \\ 
    &&&  
    -\overline{\bm{f}}\leq \bm{H}(\bm{p}_{d,\tau}+\bm{w}_{d,\tau}-\bm{l}_{d,\tau})\leq\overline{\bm{f}}:\underline{\mu}_{d,\tau},\overline{\mu}_{d,\tau},\nonumber\\
    &&&\forall \tau=1,...,T\label{DAc}\\
    &&&  0 \leq \bm{p}_{d,\tau} \leq \overline{\bm{p}},\forall \tau=1,...,T \label{DAd}\\
    &&& -\overline{\bm{r}} \leq \bm{p}_{d,\tau}-\bm{p}_{d,\tau-1} \leq \overline{\bm{r}},\forall \tau=2,...,T\label{DAe}\\
    &&& 0 \leq \bm{w}_{d,\tau} \leq \hat{\bm{y}}_{d,\tau},\forall \tau=1,...,T, \label{DAf}
\end{alignat}
\end{subequations}
where $\bm{x}_d=[\bm{x}_{d,\tau}]_{\tau=1}^T=[\bm{p}_{d,\tau}\in \mathbb{R}^N;\bm{w}_{d,\tau}\in \mathbb{R}^N]_{\tau=1}^T$ is the collection of DA decision variables, and $N$ is the number of nodes in the system. The optimal solution is denoted as $\bm{x}_d^*=[\bm{x}_{d,\tau}^*]_{\tau=1}^T=[\bm{p}_{d,\tau}^*;\bm{w}_{d,\tau}^*]_{\tau=1}^T$. $\bm{\rho}\in \mathbb{R}^N$ is the marginal cost vector of  
 traditional generators. RES enters the market with zero marginal cost. Each element in the vector $\bm{p}_{d,\tau}$ represents the power generated by a traditional generator unit, whose marginal cost is in the corresponding element in $\bm{\rho}$.
 % In the case where generators submit stepwise marginal cost curves (which is an approximation of the linear marginal cost curve \cite{kirschen2018fundamentals}), the elements in $\bm{p}_{d,\tau}$ represent different generation blocks with different costs.
 Here, we assume there is no power loss on lines, and include a DC representation of the network. The equality constraint \eqref{DAb} enforces the power balance conditions. For simplification and considering the high accuracy of demand forecast, the demand $\bm{l}_{d,\tau}\in \mathbb{R}^N$ is considered to be known with certainty.
 % \textcolor{blue}{But we note that this simplification can be easily removed by joint forecasting of demand and RES.} 
 The inequality constraints \eqref{DAc} restrict the scheduled power flow within the line flow limits. $\bm{H}$ in \eqref{DAc} is the Power Transfer Distribution Factors \cite{zimmerman2010matpower} mapping the nodal power injection to the power flow on lines. \eqref{DAd} and \eqref{DAe} are the output power and ramping limits of the traditional generators. \eqref{DAf} limits the DA schedule of RES up to the forecast $\hat{\bm{y}}_{d,\tau}$ which represents a single-value estimate of the RES production $Y_{d,\tau}$. $Y_{d,\tau}$ is a random variable since the RES production is unknown in DA market clearing.
 % After solving \eqref{DA}, the optimal solutions are obtained and denoted as $\{\bm{p}^*_{d,\tau},\bm{w}^*_{d,\tau}\}_{\tau=1}^T$.

\begin{remark}
 In line with European practices \cite{morales2023prescribing,hermann2022complementarity}, we do not incorporate binary decisions regarding unit commitment (UC) in \eqref{DA}. However, we note that UC is a requisite consideration in the U.S. markets. To analyze the market behavior, the relaxed UC problem, where the binary commitment decisions are substituted with continuous ones, is widely used \cite{kazempour2017value,zhao2022uncertainty}. The exact relaxation can be achieved by accurately adding cutting planes \cite{ferber2020mipaal}. In this way, the mixed integer program is equivalently transformed into a linear program. A common practice is to train the forecasting model using the relaxed UC formulation \cite{zhao2022uncertainty}. During the operational phase, the generated forecasts are then used as inputs to the UC problem.  
\end{remark}

Since the RES production is uncertain in DA, the DA schedules are to be adjusted at each time-slot $\tau,\forall \tau=1,...,T$ in RT on day $d$, after the RES realization $\bm{y}_{d,\tau}$ is observed. The RT market deals with the imbalance $\bm{y}_{d,\tau}-\bm{w}_{d,\tau}^*$ caused by RES, with a minimized imbalance cost. Additionally, the RT market clearing at time-slot $\tau,\forall \tau=2,...,T$ is influenced not only by the DA clearing outcomes but also by the extent of power adjustments made in the preceding time-slot. This is due to the ramping constraints that interconnect adjacent time slots. Because the market clearing is conducted separately for each day, the RT market clearing at time-slot $\tau=1$ on day $d$ remains unaffected by any adjustments made at time-slot $\tau=T$ on the previous day $d-1$. In the following, we firstly give the mathematical formulation of the RT clearing at $\tau=1$,
\begin{subequations}\label{RT}
\begin{alignat}{2} &\mathop{\min}_{\bm{z}_{d,\tau}}&&\ \bm{\rho}_{+}^\top\bm{p}_{d,\tau}^+-\bm{\rho}_{-}^\top\bm{p}^-_{d,\tau}\label{RTa}\\ 
    & \text{s.t.} && \bm{1}^\top(\bm{p}_{d,\tau}^+-\bm{p}^-_{d,\tau}-\bm{\kappa}_{d,\tau})=-\bm{1}^\top(\bm{y}_{d,\tau}-\bm{w}^*_{d,\tau})\label{RTb}
    \\ 
    &&& 
    -\overline{\bm{f}}-\bm{H}(\bm{p}_{d,\tau}^*+\bm{w}^*_{d,\tau}-\bm{l}_{d,\tau})\leq \bm{H}(\bm{p}_{d,\tau}^+-\bm{p}_{d,\tau}^--\bm{\kappa}_{d,\tau}\nonumber\\
    &&& \qquad +\bm{y}_{d,\tau}-\bm{w}^*_{d,\tau})\leq\overline{\bm{f}}-\bm{H}(\bm{p}_{d,\tau}^*+\bm{w}^*_{d,\tau}-\bm{l}_{d,\tau})\label{RTc}\\
    &&&  \bm{0} \leq \bm{p}_{d,\tau}^+ \leq \overline{\bm{p}^+} \label{RTd}\\
    &&&  \bm{0} \leq \bm{p}_{d,\tau}^- \leq \overline{\bm{p}^-} \label{RTe}\\
    &&&  0 \leq \bm{p}_{d,\tau}^*+\bm{p}_{d,\tau}^+-\bm{p}_{d,\tau}^- \leq \overline{\bm{p}} \label{RTf}\\
    &&& \bm{0} \leq \bm{\kappa}_{d,\tau} \leq \bm{y}_{d,\tau}\label{RTg}
\end{alignat}
\end{subequations}
where $\bm{z}_{d,\tau}=[\bm{p}_{d,\tau}^+\in \mathbb{R}^N;\bm{p}_{d,\tau}^-\in \mathbb{R}^N;\bm{\kappa}_{d,\tau}\in \mathbb{R}^N]$ is the collection of RT decision variables. The output power of generators may be increased by an amount $\bm{p}_{d,\tau}^+$  with the marginal cost $\bm{\rho}_+>0$ for up-regulation, or decreased by an amount $\bm{p}_{d,\tau}^-$ with the marginal utility $\bm{\rho}_->0$ for down-regulation. $\bm{\kappa}_{d,\tau}$ is the amount of RES spill. These decisions are driven by the need to settle the RES deviation $\bm{y}_{d,\tau}-\bm{w}_{d,\tau}^*$ in \eqref{RTb}. \eqref{RTc} is the power flow constraint, whose lower and upper bounds are determined by subtracting the power flow in the DA market from the line capacity. \eqref{RTd} and \eqref{RTe} limit the amount of up-regulation and down-regulation power to $\overline{\bm{p}^+},\overline{\bm{p}^-}$. For inflexible generators that cannot be dispatched in RT, the corresponding elements in the upper bounds will be zero, resulting in zero up- and down-adjustments for those generators. Additionally, the eventual generation power, considering the DA schedule $\bm{p}_{d,\tau}^*$ and the adjustment, should be within the output power limits, as stated in \eqref{RTf}. The inclusion of RES spill accounts for situations where the actual RES generation surpasses the scheduled amount in the DA schedule $\bm{w}_{d,\tau}^*$, and the excess cannot be entirely offset by the down-regulation power available from flexible generators. The amount of RES spill $\bm{\kappa}_{d,\tau}$ can be at most to its realization $\bm{y}_{d,\tau}$, as stated in \eqref{RTg}.

The RT clearing at time-slot $\tau,\forall \tau=2,...,T$ is,
\begin{subequations}\label{RT2}
\begin{alignat}{2} &\mathop{\min}_{\bm{z}_{d,\tau}}&&\ \eqref{RTa}\\ 
    & \text{s.t.} && \eqref{RTb},\eqref{RTc},\eqref{RTd},\eqref{RTe},\eqref{RTf},\eqref{RTg}\\
    &&& -\overline{\bm{r}} \leq \bm{p}_{d,\tau}^*+\bm{p}_{d,\tau}^+-\bm{p}_{d,\tau}^- -\nonumber\\ &&& \qquad (\bm{p}_{d,\tau-1}^*+\bm{p}_{d,\tau-1}^{+*}-\bm{p}_{d,\tau-1}^{-*})\leq \overline{\bm{r}}\label{RTh}
\end{alignat}
\end{subequations}

The difference between the eventual generation at time-slot $\tau-1$, denoted by $\bm{p}_{d,\tau-1}^*+\bm{p}_{d,\tau-1}^{+*}-\bm{p}_{d,\tau-1}^{-*}$, and the eventual generation at time-slot $\tau$ must satisfy the ramping constraints, as stated in \eqref{RTh}.

In this work, we follow the practice in \cite{morales2014electricity} and assume that the DA and RT markets share the same temporal granularity, i.e., one hour, leading to $T=24$. However, the market clearing model in \eqref{DA}, \eqref{RT}, and \eqref{RT2} can be easily adapted to DA and RT markets with different temporal granularity.

To obtain the unique primal and dual solutions from DA and RT clearing, we require each element in the marginal cost vectors $\bm{\rho},\bm{\rho}_+,\bm{\rho}_-$ are different. After solving the DA and RT market clearing, the eventual generation of the generators is either $\bm{p}_{d,\tau}^*+\bm{p}_{d,\tau}^{+*}$ when the RES falls short of its scheduled production in RT, or $\bm{p}_{d,\tau}^*-\bm{p}_{d,\tau}^{-*}$ when the RES generates more power than the schedule. Here, we define overall generation cost or the negative social surplus in a day.

\begin{definition}
    We define overall generation cost in a day $d$ as,
\begin{subequations}
    \begin{align}
       &\sum_{\tau=1}^T \bm{\rho}^\top \bm{p}^*_{d,\tau}+\bm{\rho}_+^\top \bm{p}^{+*}_{d,\tau}-\bm{\rho}_-^\top \bm{p}^{-*}_{d,\tau}\label{overallcost}\\
       =&\sum_{\tau=1}^T \bm{\rho}^\top (\bm{p}^*_{d,\tau}+\bm{p}^{+*}_{d,\tau})+(\bm{\rho}_+-\bm{\rho})^\top\bm{p}^{+*}_{d,\tau}+\nonumber\\
       &\sum_{\tau=1}^T\bm{\rho}^\top (\bm{p}^*_{d,\tau}-\bm{p}^{-*}_{d,\tau})+(\bm{\rho}-\bm{\rho}_-)^\top\bm{p}^{-*}_{d,\tau}
    \end{align}   \end{subequations}   
\end{definition}

$\bm{\rho}_+-\bm{\rho}$ and $\bm{\rho}-\bm{\rho}_-$ are the incremental bidding price, which reflects the marginal opportunity loss for up- and down-regulation \cite{zavala2017stochastic}. We require them to be positive for penalizing the RT adjustment due to forecasting errors. In this way, any RT adjustment would bring the extra cost either $(\bm{\rho}_+-\bm{\rho})^\top\bm{p}^{+*}_{d,\tau}$ or $(\bm{\rho}-\bm{\rho}_-)^\top\bm{p}^{-*}_{d,\tau}$. The incremental bidding prices of supplying upward and downward balancing power are usually different. This explains why forecasting the expectation hardly works well in reducing the overall cost, as it overlooks the typical asymmetry affecting the RT cost.

\subsection{Mathematical Reformulation}\label{Mathematical Reformulation}

In this subsection, we first convert the RT clearing in \eqref{RT} and \eqref{RT2} into a mathematically equivalent form, and then give the compact form of DA and RT market clearing.

We reformulate the RT clearing in \eqref{RT}. To show the upper and lower bounds of the power adjustment more clearly, we divide the constraint \eqref{RTf} into two parts. Concretely,  when the RES produces less power than the schedule, we have $\bm{p}_{d,\tau}^+ \geq \bm{0},\bm{p}_{d,\tau}^- = \bm{0}$. Conversely, when the RES produces more power than the schedule, we have $\bm{p}_{d,\tau}^- \geq \bm{0},\bm{p}_{d,\tau}^+ = \bm{0}$. We divide \eqref{RTf} into the following two constraints by the two cases,
\begin{subequations}\label{reformuRT1}
    \begin{align}
    &-\bm{p}_{d,\tau}^* \leq \bm{p}_{d,\tau}^+ \leq \overline{\bm{p}}-\bm{p}_{d,\tau}^*\\ 
    & \bm{p}_{d,\tau}^*-\overline{\bm{p}} \leq \bm{p}_{d,\tau}^- \leq \bm{p}_{d,\tau}^*
\end{align}
\end{subequations}

Since $0 \leq \bm{p}_{d,\tau}^* \leq \overline{\bm{p}}$, the left side of \eqref{reformuRT1} is less than 0. Considering the power adjustment  $\bm{p}_{d,\tau}^+,\bm{p}_{d,\tau}^-$ is larger than 0 as stated in \eqref{RTd} and \eqref{RTe}, \eqref{reformuRT1} can be further simplified as,
\begin{subequations}\label{reformuRT1simp}
    \begin{align}
    &\bm{0} \leq \bm{p}_{d,\tau}^+ \leq \overline{\bm{p}}-\bm{p}_{d,\tau}^*\\ 
    & \bm{0} \leq \bm{p}_{d,\tau}^- \leq \bm{p}_{d,\tau}^*
\end{align}
\end{subequations}
The RT market clearing at time $\tau=1$ becomes,
\begin{subequations}\label{RTreform}
\begin{alignat}{2}
&\mathop{\min}_{\bm{z}_{d,\tau}}&&\ \eqref{RTa}\\
& \text{s.t.} && \eqref{RTb},\eqref{RTc},\eqref{RTd},\eqref{RTe},\eqref{reformuRT1simp},\eqref{RTg}
\end{alignat}
\end{subequations}

Likewise, \eqref{RTh} can be equivalently reformulated as,
\begin{subequations}\label{reformuRT2simp}
    \begin{align}
    &\bm{0} \leq \bm{p}_{d,\tau}^+ \leq \overline{\bm{r}}+(\bm{p}_{d,\tau-1}^*+\bm{p}_{d,\tau-1}^{+*}-\bm{p}_{d,\tau-1}^{-*})-\bm{p}_{d,\tau}^*\\ 
    & \bm{0} \leq \bm{p}_{d,\tau}^- \leq \bm{p}_{d,\tau}^*-(\bm{p}_{d,\tau-1}^*+\bm{p}_{d,\tau-1}^{+*}-\bm{p}_{d,\tau-1}^{-*})+\overline{\bm{r}}
\end{align}
\end{subequations}

The RT market clearing at time $\tau=2,...,T$ becomes,
\begin{subequations}\label{RTreform2}
\begin{alignat}{2}
&\mathop{\min}_{\bm{z}_{d,\tau}}&&\ \eqref{RTa}\\
& \text{s.t.} && \eqref{RTb},\eqref{RTc},\eqref{RTd},\eqref{RTe},\eqref{reformuRT1simp},\eqref{reformuRT2simp},\eqref{RTg}
\end{alignat}
\end{subequations}

Next, we convert the DA market clearing in \eqref{DA} and the RT market clearing in \eqref{RTreform} and \eqref{RTreform2}, which are linear programs, into equivalent compact forms, with the dual variable listed after the colon. Concretely, the compact DA market clearing is,
\begin{subequations}\label{DAcompact}
\begin{alignat}{2}
 \bm{x}_d^*=&\mathop{\arg\min}_{\bm{x}_d} && \quad \bm{\rho}_{\text{DA}}^{\top}  \bm{x}_{d}\\
    &\text{s.t.} &&  \quad  \bm{G}_{\text{DA}}\bm{x}_d\leq \bm{\psi}_{\text{DA},d}+\bm{F}_{\text{DA}}^y\hat{\bm{y}}_d:\bm{\sigma}_{d}\label{DAcompat_b}.
\end{alignat}
\end{subequations}
The coefficients $\bm{\rho}_{\text{DA}},\bm{G}_{\text{DA}},\bm{\psi}_{\text{DA},d},\bm{F}_{\text{DA}}^y$ are constant. The RES forecasts and the demand are summarized into vectors $\hat{\bm{y}}_d=[\hat{\bm{y}}_{d,\tau}]_{\tau=1}^T$ and $\bm{l}_d=[\bm{l}_{d,\tau}]_{\tau=1}^T$, respectively.
The value of $\bm{\psi}_{\text{DA},d}$ varies from day to day due to its dependence on the demand $\bm{l}_d$. Likewise, the RT market clearing in \eqref{RTreform} and \eqref{RTreform2} are converted into the compact forms,
\begin{subequations}\label{RTcompact}
\begin{alignat}{2}
 \bm{z}_{d,\tau}^*=&\mathop{\arg\min}_{\bm{z}_{d,\tau}} && \quad \bm{\rho}_{\text{RT}}^{\top}  \bm{z}_{d,\tau}\\
    &\text{s.t.} &&  \bm{G}_{\text{RT}}\bm{z}_{d,\tau}\leq \bm{\psi}_{\text{RT},d,\tau}+\bm{F}_{\text{RT}}\bm{x}_{d,\tau}^*:\bm{\nu}_{d,\tau}, \tau=1\label{RTcompat_b}.
\end{alignat}
\end{subequations}
\begin{subequations}\label{RTcompact2}
\begin{alignat}{2}
 \bm{z}_{d,\tau}^*=&\mathop{\arg\min}_{\bm{z}_{d,\tau}} && \quad \bm{\rho}_{\text{RT}}^{\top}  \bm{z}_{d,\tau}\\
    &\text{s.t.} &&  \quad  \bm{G}_{\text{RT}}^\prime\bm{z}_{d,\tau}\leq \bm{\psi}_{\text{RT},d,\tau}^\prime+[\bm{F}_{\text{RT}}^{\prime x},\bm{F}_{\text{RT}}^{\prime p},\bm{F}_{\text{RT}}^{\prime +-}]\nonumber\\
&&& [\bm{x}^*_{d,\tau},\bm{p}^*_{d,\tau-1},\bm{p}^{+-*}_{d,\tau-1}]^\top: \bm{\zeta}_{d,\tau},\forall \tau=2,...,T\label{RTcompat2_b}.
\end{alignat}
\end{subequations}
 The coefficients $\bm{\rho}_{\text{RT}},\bm{G}_{\text{RT}},\bm{\psi}_{\text{RT},d,\tau},\bm{F}_{\text{RT}}$ and $\bm{G}^\prime_{\text{RT}},\bm{\psi}^\prime_{\text{RT},d,\tau},\bm{F}^{\prime x}_{\text{RT}},\bm{F}^{\prime p}_{\text{RT}},\bm{F}^{\prime +-}_{\text{RT}}$ are constant. The values of $\bm{\psi}_{\text{RT},d,\tau},\bm{\psi}_{\text{RT},d,\tau}^\prime$ vary from hour to hour due to the dependence on the RES realization $\bm{y}_{d,\tau}$. The parameter $\bm{p}^{+-*}_{d,\tau-1}=[\bm{p}^{+*}_{d,\tau-1};\bm{p}^{-*}_{d,\tau-1}]$ contains RT power adjustment for up- and down-regulation at previous time $\tau-1$.

The market clearing model developed in this work accounts for the participation of flexible conventional generators and RES on the supply side. This modeling approach is consistent with those used in \cite{morales2014electricity, zhao2022uncertainty}. Although the discussions in the following sections are based on the DA and RT models given in \eqref{DA}, \eqref{RT}, and \eqref{RT2}, we emphasize that even when additional participants such as energy storage are included, provided the model remains linear and the structure of the DA and RT market clearing problems still aligns with the formulations in \eqref{DAcompact}, \eqref{RTcompact}, and \eqref{RTcompact2}, the proposed method for deriving the loss function remains applicable. In the Appendix \ref{compat}, we use energy storage as an example to show that the structure of the market clearing model with energy storage aligns with \eqref{DAcompact}, \eqref{RTcompact} \eqref{RTcompact2}.

% \vspace{-2em}
\section{Methodology}\label{Parameter Estimation}

\begin{figure}[t]
  \centering
  \includegraphics[scale=0.45]{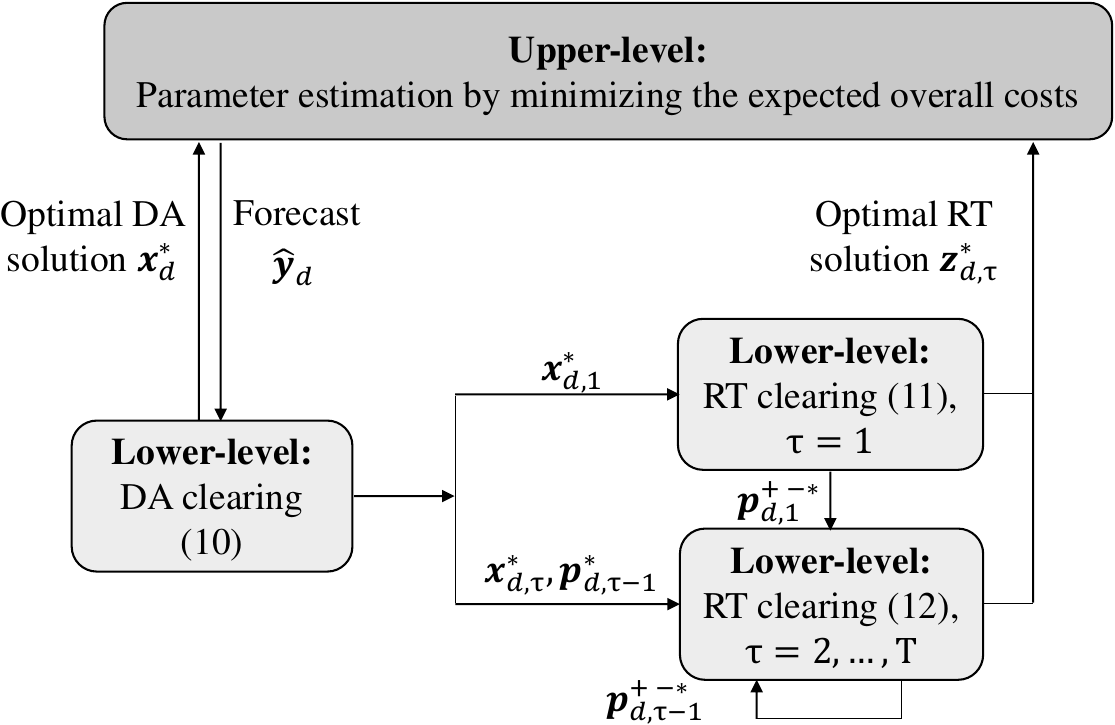}
  \caption{An illustration of the parameter estimation problem.}
\label{bilevelV2}
\end{figure}

% The forecast $\hat{\bm{y}}_{d,\tau}$ in \eqref{DA} determines the RES entering quantities in the DA market.
In this section, we show how to identify the RIEQ, whose value corresponds to $\hat{\bm{y}}_{d,\tau}$, an input parameter in \eqref{DA}, via training forecasting models. Let $g(\ \cdot\ ;\Theta)$ denote the forecasting model with the parameter $\Theta$, and $\bm{s}_{d,\tau}$ denote the context. The RIEQ forecast for the time $\tau$ on day $d$ is,
\begin{equation}\label{fm}
    \hat{\bm{y}}_{d,\tau}= g(\bm{s}_{d,\tau};\Theta)
\end{equation}
where $\hat{\bm{y}}_{d,\tau}$ denotes the forecast. 

Typically, there are two phases, i.e., the training phase and the operational phase. During the training phase, the parameter $\Theta$ is estimated using historical data by solving an optimization problem in which the loss function serves as the objective. When a neural network is used as the forecasting model, this estimation involves forward computation followed by backward gradient descent. At the operational phase, forecasts are generated using the estimated parameters through a computationally efficient forward pass..

In the following, we first formulate the training phase as a bilevel optimization problem \cite{zhang2024valueoriented} to estimate the parameters of the forecasting model, followed by the operational phase. We then discuss key market properties under the issued RIEQ.

\subsection{Training Phase}
Data in the training set $\{\{\bm{s}_{d,\tau},y_{d,\tau}\}_{\tau=1}^T\}_{d=1}^D$ is available, which consists of historical context and RES realization in $D$ days. An illustration of the bilevel program is shown in Fig. \ref{bilevelV2}. The objective of determining RIEQ is to minimize the overall DA and RT operating costs while accounting for the impact of DA market decisions made under uncertainty on RT market clearing, which corrects power imbalances caused by RES. As a result, even with a deterministic DA clearing, the impact of RES uncertainty on both DA and RT markets is accounted for through RIEQ, ensuring coordination between the two markets. The upper level determines the model parameter $\Theta$, while the lower level involves the DA and RT market clearings. The bilevel program is mathematically formulated as,
\begin{subequations}\label{bilevel}
\begin{alignat}{2} 
&\underset{{\Theta}}{\mathop{\min}}&&  \frac{1}{D\cdot T}\sum_{d=1}^{D}\{\bm{\rho}_{\text{DA}}^\top\bm{x}_d^*+\sum_{\tau=1}^{T}\bm{\rho}_{\text{RT}}^\top\bm{z}_{d,\tau}^*\}\label{bilevel_a}\\
    & \text{s.t.} 
    && \hat{\bm{y}}_{d,\tau}=g(\bm{s}_{d,\tau};\Theta),\forall \tau=1,...,T,\forall d=1,...,D\\ 
    &&&0 \leq \hat{\bm{y}}_{d,\tau} \leq \bar{\bm{y}}_{d,\tau},\forall \tau=1,...,T,\forall d=1,...,D\label{bilevel_c}\\
    &&&\begin{rcases}\eqref{DAcompact},\forall d=1,...,D\label{bilevel_d}\\
    \eqref{RTcompact}, \forall d=1,...,D\\ \eqref{RTcompact2},\forall d=1,...,D\end{rcases}\text{Lower level}
\end{alignat}
\end{subequations}
where the upper-level objective \eqref{bilevel_a} seeks to minimize the expected overall operating cost of the two markets. This is achieved by leveraging the optimal DA and RT cost functions, which are informed by the decisions obtained from the lower level \eqref{bilevel_d}. \eqref{bilevel_c} limits the forecast $\hat{\bm{y}}_{d,\tau}$ within $\bar{\bm{y}}_{d,\tau}$, which can be RES capacity. The lower level treats the forecast $\hat{\bm{y}}_{d,\tau}$ as an input parameter. As a consequence, both DA and RT decisions are affected by it.

To show the impact of the forecast on the operating cost more clearly, we replace the lower level with the dual problems. The overall operating cost within the upper-level objective is then substituted with the DA and RT dual objectives. These objectives are constructed as a linear combination of the right-side parameters and the associated dual variables, i.e,
\begin{subequations}\label{bileveldual}
\begin{alignat}{2} 
&\underset{{\Theta}}{\mathop{\min}}&&  \frac{1}{D\cdot T}\sum_{d=1}^{D}\{-\bm{\sigma}_{d}^{*\top}(\bm{\psi}_{\text{DA},d}+\bm{F}_{\text{DA}}^y\hat{\bm{y}}_d)\nonumber\\
&&&-\bm{\nu}_{d,1}^{*\top}(\bm{\psi}_{\text{RT},d,1}+\bm{F}_{\text{RT}}\bm{x}_{d,1}^*)+\nonumber\\
&&&\sum_{\tau=2}^{T}-\bm{\zeta}_{d,\tau}^{*\top}(\bm{\psi}_{\text{RT},d,\tau}^\prime+\bm{F}_{\text{RT}}^{\prime x}\bm{x}_{d,\tau}^*+\bm{F}_{\text{RT}}^{\prime p}\bm{p}_{d,\tau-1}^*+\nonumber\\
&&&\bm{F}_{\text{RT}}^{\prime +-}\bm{p}_{d,\tau-1}^{+-*})\}\label{bileveldual_a}\\
    & \text{s.t.} 
    && \hat{\bm{y}}_{d,\tau}=g(\bm{s}_{d,\tau};\Theta),\forall \tau=1,...,T,\forall d=1,...,D\\ 
    &&&0 \leq \hat{\bm{y}}_{d,\tau} \leq \bar{\bm{y}}_{d,\tau},\forall \tau=1,...,T,\forall d=1,...,D\label{bileveldual_c}\\
    &&&\bm{\sigma}_{d}^*=\mathop{\arg\max}_{\bm{\sigma}_{d}\geq0}-\bm{\sigma}_{d}^{\top}(\bm{\psi}_{\text{DA},d}+\bm{F}_{\text{DA}}^y\hat{\bm{y}}_d),\forall d=1,...,D\label{bileveldual_d}\\
    &&&\eqref{DAcompact},\forall d=1,...,D\label{bileveldual_e}\\
    &&&\bm{\nu}_{d,1}^*=\mathop{\arg\max}_{\nu_{d,1} \geq 0}-\bm{\nu}_{d,1}^{\top}(\bm{\psi}_{\text{RT},d,1}+\bm{F}_{\text{RT}}\bm{x}_{d,1}^*),\forall d=1,...,D\label{bileveldual_f}\\
    &&&\eqref{RTcompact},\forall d=1,...,D\label{bileveldual_g}\\
    &&& \bm{\zeta}_{d,\tau}^*=\mathop{\arg\max}_{\bm{\zeta}_{d,\tau} \geq 0}-\bm{\zeta}_{d,\tau}^{\top}(\bm{\psi}_{\text{RT},d,\tau}^\prime+\bm{F}_{\text{RT}}^{\prime x}\bm{x}_{d,\tau}^*+\bm{F}_{\text{RT}}^{\prime p}\bm{p}_{d,\tau-1}^*+\nonumber\\
&&&\bm{F}_{\text{RT}}^{\prime +-}\bm{p}_{d,\tau-1}^{+-*}),\forall \tau=2,...,T,\forall d=1,...,D\label{bileveldual_h}\\
&&&\eqref{RTcompact2},\forall d=1,...,D\label{bileveldual_i}
\end{alignat}
\end{subequations}
where \eqref{bileveldual_d} is the dual problem of DA clearing. \eqref{bileveldual_f} is the dual problem of RT clearing at $\tau=1$, and \eqref{bileveldual_h} is the dual problem of RT clearing at $\tau=2,...,T$.
Since RT clearing requires the primal solutions of DA clearing and the previous RT clearing as input parameters, we also include the primal problems in the lower level. The forecast $\hat{\bm{y}}_d$ affects the upper-level objective \eqref{bileveldual_a} via its impact on the DA and RT dual solutions  $\bm{\sigma}_d^*,\bm{\nu}_{d,1}^*,\bm{\zeta}_{d,\tau}^*$, and their primal solutions $\bm{x}^*_{d,\tau},\bm{p}^*_{d,\tau-1},\bm{p}^{+-*}_{d,\tau-1}$. If we can obtain the function between them and the forecast $\hat{\bm{y}}_d$ directly, the upper-level objective can be rewritten as a function regarding the forecast $\hat{\bm{y}}_d$, and can be used as the loss function for training. The specific design of the loss function is deferred to Section \ref{loss function design}.

We note that the training is conducted offline. That is, it is performed using historical data from DA and RT markets, including demand, RES realizations, and generator bidding information. No actual market clearing occurs during the training phase. 

\subsection{Operational Phase}
At the operational phase, the forecast $\bm{\hat{y}}_{d,\tau}$ under the context $\bm{s}_{d,\tau}$ can be obtained by the trained model. Subsequently, utilizing this forecast, the operator proceeds to solve the DA clearing in \eqref{DA}, and obtains the DA schedules of traditional generator and RES, i.e., $\bm{p}_{d,\tau}^*,\bm{w}_{d,\tau}^*,\forall \tau=1,...,T$. After the RES realization is revealed at each time $\tau=1,\ldots,T$, the RT market described in \eqref{RT} and \eqref{RT2} is cleared to resolve the power imbalance between the realized RES output and its scheduled value, i.e., $\bm{1}^\top(\bm{y}_{d,\tau}-\bm{w}_{d,\tau}^*)$.

\subsection{Market Properties}

In this work, we assume that the DA market is primarily affected by uncertainty, while the RT market is cleared based on RES realization or a highly accurate RT RES forecast close to the realization. To improve the coordination between DA and RT markets, RIEQ appears as the DA RES forecast $\hat{\bm{y}}_{d,\tau}$ in \eqref{DAf}. Compared to the existing DA market clearing model that uses traditional DA RES forecasts, such as the expected RES generation, the market clearing structure remains unchanged and the values of RES forecasts is the only difference. In this regard, the current DA market properties are preserved, i.e.,

\textit{1) Cost recovery}: The profits of market players are nonnegative at the market clearing solutions. As the DA market clearing in \eqref{DA} determines the schedules of traditional generators and RES, cost recovery specifically ensures that their profits remain nonnegative given the market clearing schedules.

\textit{2) Revenue adequacy}: Revenue adequacy implies that, at the market clearing solutions, the total payment made by the load to the market operator equals the total payment made by the market operator to traditional and renewable generators, as well as to the transmission line operator. 

The proof for the above properties is provided in the Appendix \ref{proof_property}. Note that these properties can only be ensured in expectation for stochastic market clearing, which is a key drawback.

\section{Loss Function Design}
\label{loss function design}
We derive the loss function based on the bilevel program in \eqref{bileveldual}. We first analyze how forecasts influence the DA and RT optimal solutions. Based on it, a loss function is derived, and will be used for training a forecasting model.

\subsection{The Impact of the Forecast on the DA and RT Solutions}
This section analyzes the impact of the forecasts on the primal and dual solutions of DA and RT market clearings. Specifically,  we derive analytical functions that quantitatively depict the impact of the forecasts on the optimal primal and dual solutions, whose illustration is shown in Fig. \ref{chain}. As a parameter to the DA market clearing \eqref{DAcompact}, forecasts $\hat{\bm{y}}_d$ influence DA  primal and dual solutions directly. As for the RT market clearing, the impact of the forecast is indirect, as it does not directly appear as the parameter in \eqref{RTcompact} or \eqref{RTcompact2}. Concretely, for the RT market clearing at time-slot $\tau=1$, the forecast $\hat{\bm{y}}_d$ influences  $\bm{x}_{d,1}^*$, and then $\bm{x}_{d,1}^*$ influences the RT solutions, with the first impact being determined by the DA clearing \eqref{DAcompact}, and the second impact by the RT clearing \eqref{RTcompact}. For the RT clearing at time-slot $\tau=2,...,T$, the impact of the forecast is more complex. The forecast $\hat{\bm{y}}_d$ affects the DA solutions $\bm{x}_{d,\tau}^*,\bm{p}_{d,\tau-1}^*$ through the DA clearing \eqref{DAcompact}. The influence of the forecast $\hat{\bm{y}}_d$ on the RT solutions $\bm{p}_{d,\tau-1}^{+-*}$ at the previous time-slot $\tau-1$ occurs as explained earlier. Then, the influence of the parameters $\bm{x}_{d,\tau}^*,\bm{p}^*_{d,\tau-1},\bm{p}^{+-*}_{d,\tau-1}$  on the RT clearing solutions at time $\tau$ is determined by \eqref{RTcompact2}. 

\begin{figure*}[t]
\centering
\includegraphics[scale=0.6]{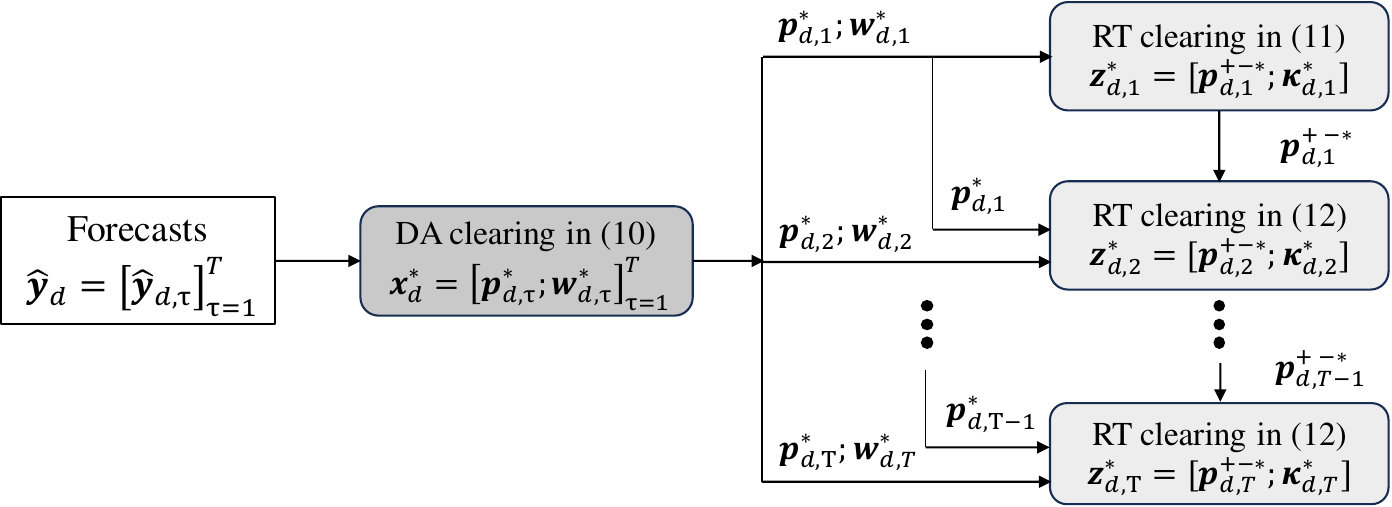}
\caption{An illustration of the impact of forecasts on DA and RT primal solutions.}
  
\label{chain}
\end{figure*}

Since the core of the above analysis is understanding the impact of the parameters on the optimization solutions, we use the multiparametric theory for this end. A general linear program \eqref{lp} is used as an example. We firstly define primal and dual decision policies. Then, the theorem regarding them is presented.
\begin{subequations}\label{lp}
\begin{alignat}{2}
 \bm{x}^*=&\mathop{\arg\min}_{\bm{x}} && \quad \bm{c}^\top  \bm{x}\label{lpa}\\
    &\text{s.t.} &&  \quad  \bm{G}\bm{x}\leq \bm{\psi}+\bm{F}\bm{\omega}:\bm{\sigma}\label{lpb}.
\end{alignat}
\end{subequations}
\begin{definition}
\emph{\textbf{(Primal and dual decision policies)}}
Primal and dual decision policies are functions defined across the polyhedral set $\Omega$, which describe the change in the optimal primal and dual solutions, i.e., $\bm{x}^*$ and $ \bm{\sigma}^*$, as the parameter $\bm{\omega}$ varies in $\Omega$.

\end{definition}

\begin{theorem} \label{theorem1}
\emph{\textbf{(\cite{borrelli2003geometric})}} Consider the linear program \eqref{lp} and the parameter $\bm{\omega} \in \Omega$. The primal and dual decision policies are a piecewise linear function and a stepwise function respectively, if there exists a polyhedral partition $R_1,...,R_N$ of $\Omega$, and $\forall \bm{\omega}\in R_i$, the primal decision policy is linear, and the dual decision policy is a constant function.
\end{theorem}

Theorem \ref{theorem1} implies that in a neighborhood of the parameter, the primal decision policy is represented by a linear function, whereas the corresponding dual decision policy remains a constant function. Given a specific value of $\bm{w}$, we study the local policies defined in its neighborhood. The local dual decision policy can be obtained easily, as it is a constant function. Its output is the optimal dual solution obtained by solving the dual problem of \eqref{lp}, given the value of $\bm{\omega}$. Additionally, after solving \eqref{lp}, the active constraints of \eqref{lpb} can be obtained. Let $\mathcal{J}$ denote the row index set associated with \eqref{lpb}, and $\mathcal{J}^a$ denote the row index set of the active constraints. The parameters associated with the active constraints are denoted as $\bm{G}[\mathcal{J}^a],\bm{\psi}[\mathcal{J}^a],\bm{F}[\mathcal{J}^a]$. They are the sub-matrices and sub-vectors of $\bm{G},\bm{\psi},\bm{F}$ and are comprised by the rows of $\bm{G},\bm{\psi},\bm{F}$ in the row index sets $\mathcal{J}^a$.  We have the following proposition for the local primal decision policy,
\begin{proposition}\label{prop1}
    The local primal decision policy of \eqref{lp} is,
    \begin{equation}\label{policyexample}
        \bm{x}^*=\bm{G}[\mathcal{J}^a]^{-1}(\bm{\psi}[\mathcal{J}^a]+\bm{F}[\mathcal{J}^a]\bm{\omega})
    \end{equation}
\end{proposition}
\begin{proof}
    With the active constraints of \eqref{lp}, we have $\bm{G}[\mathcal{J}^a]\bm{x}^*=\bm{\psi}[\mathcal{J}^a]+\bm{F}[\mathcal{J}^a]\bm{\omega}$.  $\bm{G}[\mathcal{J}^a]^{-1}$ is the pseudo inverse of $\bm{G}[\mathcal{J}^a]$. By moving $\bm{G}[\mathcal{J}^a]$ from left to right, we have \eqref{policyexample}.
\end{proof}

\begin{remark}
\eqref{policyexample} calculates the inverse of  $\bm{G}[\mathcal{J}^a]$, whose computational complexity depends on the matrix size, i.e., $|\mathcal{J}^a|$. We note that a similar matrix inverse is also involved in \cite{donti2017task}, with the matrix size of $|\mathcal{J}|+\iota$, where $\iota$ is the dimension of $\bm{x}$. There is $|\mathcal{J}|+\iota > |\mathcal{J}^a|$. Therefore, the computation burden of \cite{donti2017task} is larger.
\end{remark}

With Proposition \ref{prop1}, we present the local primal decision policy for the DA clearing \eqref{DAcompact}, and RT clearing \eqref{RTcompact} and \eqref{RTcompact2}. Let $\mathcal{J}^a_{\text{DA},d},\mathcal{J}^a_{\text{RT},d,\tau},\forall \tau=1,...,T$ denote the row index set of active constraints \eqref{DAcompat_b}, \eqref{RTcompat_b}, and \eqref{RTcompat2_b}. The parameters associated with the active constraints are denoted as $\bm{G}_{\text{DA}}[\mathcal{J}_{\text{DA},d}^a],\bm{\psi}_{\text{DA},d}[\mathcal{J}_{\text{DA},d}^a],\bm{F}_{\text{DA}}^y[\mathcal{J}_{\text{DA},d}^a]$, $\bm{G}_{\text{RT}}[\mathcal{J}_{\text{RT},d,\tau}^a],\bm{\psi}_{\text{RT},d,\tau}[\mathcal{J}_{\text{RT},d,\tau}^a],\bm{F}_{\text{RT}}[\mathcal{J}_{\text{RT},d,\tau}^a]$,$\bm{G}_{\text{RT}}^\prime[\mathcal{J}_{\text{RT},d,\tau}^{ a}]$, and $\bm{\psi}_{\text{RT},d,\tau}^\prime[\mathcal{J}_{\text{RT},d,\tau}^{ a}],\bm{F}_{\text{RT}}^{\prime x}[\mathcal{J}_{\text{RT},d,\tau}^{ a}],\bm{F}_{\text{RT}}^{\prime p}[\mathcal{J}_{\text{RT},d,\tau}^{ a}],\bm{F}_{\text{RT}}^{\prime +-}[\mathcal{J}_{\text{RT},d,\tau}^{ a}]$. We have the following proposition for the local primal decision policies of \eqref{DAcompact}, \eqref{RTcompact}, and \eqref{RTcompact2}.

\begin{proposition}\label{prop2}
    The local primal decision policies of \eqref{DAcompact},\eqref{RTcompact} and \eqref{RTcompact2} are,
\begin{equation}\label{DApolicy}
        f^x_{d}(\hat{\bm{y}}_d):=\bm{x}_d^*=\bm{G}_{\text{DA}}[\mathcal{J}_{\text{DA},d}^a]^{-1}(\bm{\psi}_{\text{DA},d}[\mathcal{J}_{\text{DA},d}^a]+\bm{F}_{\text{DA}}^y[\mathcal{J}_{\text{DA},d}^a]\hat{\bm{y}}_d)
    \end{equation}
\begin{equation}\label{RTpolicy}
\begin{split}
        &\bm{z}_{d,\tau}^*=\bm{G}_{\text{RT}}[\mathcal{J}_{\text{RT},d,\tau}^a]^{-1}(\bm{\psi}_{\text{RT},d,\tau}[\mathcal{J}_{\text{RT},d,\tau}^a]+\\
        &\bm{F}_{\text{RT}}[\mathcal{J}^a_{\text{RT},d,\tau}]\bm{x}_{d,\tau}^*), \tau=1
    \end{split}
    \end{equation}
    \begin{equation}\label{RTpolicy2}
    \begin{split}
    &\bm{z}_{d,\tau}^*=\bm{G}_{\text{RT}}^{\prime}[\mathcal{J}_{\text{RT},d,\tau}^{ a}]^{-1}(\bm{\psi}^\prime_{\text{RT},d,\tau}[\mathcal{J}_{\text{RT},d,\tau}^{ a}]+\bm{F}_{\text{RT}}^{\prime x}[\mathcal{J}_{\text{RT},d,\tau}^{ a}]\bm{x}_{d,\tau}^*+\\
    &\bm{F}_{\text{RT}}^{\prime p}[\mathcal{J}_{\text{RT},d,\tau}^{ a}]\bm{p}_{d,\tau-1}^*+\bm{F}_{\text{RT}}^{\prime +-}[\mathcal{J}_{\text{RT},d,\tau}^{ a}]\bm{p}_{d,\tau-1}^{+-*}),\forall \tau=2,...,T
    \end{split} 
    \end{equation}
\end{proposition}
% \begin{proof}
%     See Appendix B.
% \end{proof}

Eq. \eqref{DApolicy} is a linear function of $\hat{\bm{y}}_d$, whose output is the DA solution $\bm{x}_d^*$ on the day $d$. Based on \eqref{DApolicy}, we will obtain the function between the forecast $\hat{\bm{y}}_d$ and DA solution $\bm{x}_{d,\tau}^*$ at each time $\tau$, denoted as $f^x_{d,\tau}(\hat{\bm{y}}_d)$, as well as the function between the forecast $\hat{\bm{y}}_d$ and DA generator schedule $\bm{p}_{d,\tau-1}^*$, denoted as $f_{d,\tau-1}^p(\hat{\bm{y}}_d), \forall \tau=2,...,T$.  \eqref{RTpolicy} quantifies the impact of the DA solution $\bm{x}_{d,\tau}^*$ on the RT solution $\bm{z}_{d,\tau}^*$ at time $\tau=1$. \eqref{RTpolicy2} quantifies the impact of the DA solution $\bm{x}_{d,\tau}^*$, DA generator schedule $\bm{p}_{d,\tau-1}^*$  and RT generator adjustment $\bm{p}_{d,\tau-1}^{+-*}$ at previous time $\tau-1$ on the RT solution $\bm{z}_{d,\tau}^*$, $\forall \tau=2,...,T$. We will derive the function between the forecast $\hat{\bm{y}}_d$ and the RT solution $\bm{z}_{d,\tau}^*, \forall \tau=1,...,T$, which is denoted as $f_{d,\tau}^z(\hat{\bm{y}}_d), \forall \tau=1,...,T$. With this function, we will obtain the function between the RT adjustment $\bm{p}_{d,\tau-1}^{+-*},\forall \tau=2,..,T$ and the forecast $\hat{\bm{y}}_d$, denoted as $f_{d,\tau-1}^{+-}(\hat{\bm{y}}_d),\forall \tau=2,..,T$. Since \eqref{DApolicy},\eqref{RTpolicy},\eqref{RTpolicy2} are linear functions, the functions derived based on them are also linear in a neighborhood. The derivation of these functions is provided in the Appendix \ref{policy_apendix}. The function between the optimal dual solutions and the forecast $\hat{\bm{y}}_d$ is a constant function, whose output can be obtained by solving the dual problems of \eqref{DAcompact},\eqref{RTcompact}, and \eqref{RTcompact2}. With these, we are ready to transform the upper-level objective \eqref{bileveldual_a} to a function of the forecast $\hat{\bm{y}}_d$. The details are in the next subsection.

\vspace{-1em}
\subsection{Surrogate Loss Function Design}
By substituting the derived functions $f_{d,\tau}^x(\hat{\bm{y}}_d)$, $f_{d,\tau-1}^p(\hat{\bm{y}}_d)$, and $f_{d,\tau-1}^{+-}(\hat{\bm{y}}_d)$  and dual solutions into the upper-level objective \eqref{bileveldual_a}, the loss function in the neighborhood of the forecast $\hat{\bm{y}}_d$ is,
\begin{equation}\label{valueloss}   \begin{aligned}
&\ell_d(\hat{\bm{y}}_d):=-\bm{\sigma}_{d}^{*\top}(\bm{\psi}_{\text{DA},d}+\bm{F}_{\text{DA}}^y\hat{\bm{y}}_d)-\bm{\nu}_{d,1}^{*\top}(\bm{\psi}_{\text{RT},d,1}+\\
&\bm{F}_{\text{RT}}f_{d,1}^x(\hat{\bm{y}}_d))+\sum_{\tau=2}^{T}-\bm{\zeta}_{d,\tau}^{*\top}(\bm{\psi}_{\text{RT},d,\tau}^\prime+\bm{F}_{\text{RT}}^{\prime x}f_{d,\tau}^x(\hat{\bm{y}}_d)+\\
&\bm{F}_{\text{RT}}^{\prime p}f_{d,\tau-1}^p(\hat{\bm{y}}_d)+\bm{F}_{\text{RT}}^{\prime +-}f_{d,\tau-1}^{+-}(\hat{\bm{y}}_d))
\end{aligned}
\end{equation}

Since the functions $f_{d,\tau}^x(\hat{\bm{y}}_d)$, $f_{d,\tau-1}^p(\hat{\bm{y}}_d)$, and $f_{d,\tau-1}^{+-}(\hat{\bm{y}}_d)$ are linear in a neighoborhood, and the dual decision policies are constant, the loss function $\ell_d(\hat{\bm{y}}_d)$ in the neighborhood of the forecast $\hat{\bm{y}}_d$ is linear, and outputs the overall operating cost \eqref{overallcost} given the forecast $\hat{\bm{y}}_d$. In this way, the forecasting objective aligns with the subsequent operating objective, namely, coordinating DA and RT markets under RES uncertainty.

Naturally, the derivative of $\ell_d(\hat{\bm{y}}_d)$ w.r.t. the forecast $\hat{\bm{y}}_d$, i.e., $\frac{\partial \ell_d(\hat{\bm{y}}_d)}{\partial \hat{\bm{y}}_d}$, measures the marginal impact of the forecast on the overall cost and is a constant. It is,
\begin{equation}\label{gradient}
\begin{split}
&(\frac{\partial \ell_d(\hat{\bm{y}}_d)}{\partial \hat{\bm{y}}_d})^\top=-\bm{\sigma}_d^{*\top}\bm{F}_{\text{DA}}^y-\bm{\nu}_{d,1}^{*\top}\bm{F}_{\text{RT}}\frac{\partial f_{d,1}^x(\hat{\bm{y}}_d)}{\partial \hat{\bm{y}}_d}+\sum_{\tau=2}^T-\bm{\zeta}_{d,\tau}^{*T}\\
&(\bm{F}_{\text{RT}}^{\prime x}\frac{\partial f_{d,\tau}^x(\hat{\bm{y}}_d)}{\partial \hat{\bm{y}}_d}+\bm{F}_{\text{RT}}^{\prime p}\frac{\partial f_{d,\tau-1}^p(\hat{\bm{y}}_d)}{\partial \hat{\bm{y}}_d}+\bm{F}_{\text{RT}}^{\prime +-}\frac{\partial f_{d,\tau-1}^{+-}(\hat{\bm{y}}_d)}{\partial \hat{\bm{y}}_d}) 
\end{split}
\end{equation}
where the derivatives $\frac{\partial f_{d,\tau}^x(\hat{\bm{y}}_d)}{\partial \hat{\bm{y}}_d}$,  $\frac{\partial f_{d,\tau-1}^p(\hat{\bm{y}}_d)}{\partial \hat{\bm{y}}_d}$, $\frac{\partial f_{d,\tau-1}^{+-}(\hat{\bm{y}}_d)}{\partial \hat{\bm{y}}_d}$ in \eqref{gradient} are constants, which measure the marginal impact of the forecast on the DA and RT primal solutions. We have the following proposition regarding the loss function.
\begin{proposition}
   The loss function defined across the entire space of $\hat{\bm{y}}_d$ is a piecewise linear function. Each piece is a linear function described in \eqref{valueloss}.
\end{proposition}
\begin{proof}
    According to Theorem \ref{theorem1}, $f_{d,\tau}^x(\hat{\bm{y}}_d)$, $f_{d,\tau}^p(\hat{\bm{y}}_d)$, and $f_{d,\tau-1}^{+-}(\hat{\bm{y}}_d)$ are piecewise linear functions defined across the entire space of $\hat{\bm{y}}_d$. Since the loss function is the linear transformation of $f_{d,\tau}^x(\hat{\bm{y}}_d)$, $f_{d,\tau}^p(\hat{\bm{y}}_d)$, and $f_{d,\tau-1}^{+-}(\hat{\bm{y}}_d)$, it is also a piecewise linear function. 
\end{proof}
\begin{remark}
Since local primal decision policies $f_{d,\tau}^x(\hat{\bm{y}}_d)$, $f_{d,\tau-1}^p(\hat{\bm{y}}_d)$, and $f_{d,\tau-1}^{+-}(\hat{\bm{y}}_d)$ are determined by active constraint index sets $\mathcal{J}^a_{\text{DA},d},\mathcal{J}^{ a}_{\text{RT},d,\tau-1},\forall \tau=2,...,T$ of DA and RT market clearing, each piece of the loss function in \eqref{valueloss} is also related with different active constraint index sets.
\end{remark}

 One way is to enumerate all possible active constraint index sets, and derive the corresponding piece of the loss function \cite{Gal+2010}. However, the enumeration can be computationally expensive, particularly when dealing with large-scale optimization. Therefore, we propose to derive the pieces of the loss function in an on-demand manner. Concretely, during training, we only recalculate the piece of the loss function when encountering new active index sets. The details are given in the next Section.

% According to Theorem \ref{theorem1}, the loss function defined across the entire space of $\hat{\bm{y}}_d$ is expected to be a piecewise linear function. Specifically, each piece is related with a different active constraint index sets $\mathcal{J}^a_{\text{DA},d},\mathcal{J}^{ a}_{\text{RT},d,\tau-1},\forall \tau=2,...,T$ of DA and RT market clearings. Such index sets determine the primal decision policies \eqref{DApolicy}-\eqref{RTpolicy2}, and the following functions in \eqref{functionxd}-\eqref{functionz2}. One way is to enumerate all possible active constraint index sets, and derive the corresponding loss function \cite{Gal+2010}. However, the enumeration can be computationally expensive, particularly when dealing with large-scale optimization. We notice that the derivatives $\frac{\partial f_{d,\tau}^x(\hat{\bm{y}}_d)}{\partial \hat{\bm{y}}_d}$,  $\frac{\partial f_{d,\tau-1}^p(\hat{\bm{y}}_d)}{\partial \hat{\bm{y}}_d}$, $\frac{\partial f_{d,\tau-1}^{+-}(\hat{\bm{y}}_d)}{\partial \hat{\bm{y}}_d}$ in \eqref{gradient} associated with the active index sets are constants, as \eqref{functionxd},\eqref{functionp},\eqref{functionadj}\label{valuelossb} are linear functions.  
% This suggests that there is no need to recalculate these derivatives when encountering the same active index sets during the training. Therefore, we propose a solution strategy, where the derivatives are recalculated only when encountering new ones. 

\section{Training Algorithm}
We illustrate the training phase of the forecasting model based on neural networks (NNs). With the loss function, we use batch optimization to train NN. Given a batch of data in $B$ days, the parameter estimation with the derived loss function is formulated as,
\begin{subequations}\label{training_nn}
\begin{alignat}{2} 
&\underset{{\Theta}}{\mathop{\min}}&&  \frac{1}{B\cdot T}\sum_{d=1}^{B}\ell_d(\hat{\bm{y}}_d)\\
    & \text{s.t.} 
    && \hat{\bm{y}}_{d,\tau}=g(\bm{s}_{d,\tau};\Theta),\forall \tau=1,...,T,\forall d=1,...,B\\ 
    &&&0 \leq \hat{\bm{y}}_{d,\tau} \leq \bar{\bm{y}}_{d,\tau},\forall \tau=1,...,T,\forall d=1,...,B\label{trainingb}
\end{alignat}
\end{subequations}

Different from the conventional unconstrained program at the training phase, \eqref{training_nn} is with the box constraint \eqref{trainingb} for the NN output. We design a specific model structure to address this. Specifically, a Sigmoid function, whose output is between 0 and 1, is used as the activation function at the output layer. By multiplying its output with the cap $\bar{\bm{y}}_{d,\tau}$ of each sample, the constraint \eqref{trainingb} is satisfied. NN's structure is in Fig. \ref{MLP}. 
\begin{figure}[!ht]
  \centering
  \includegraphics[scale=0.55]{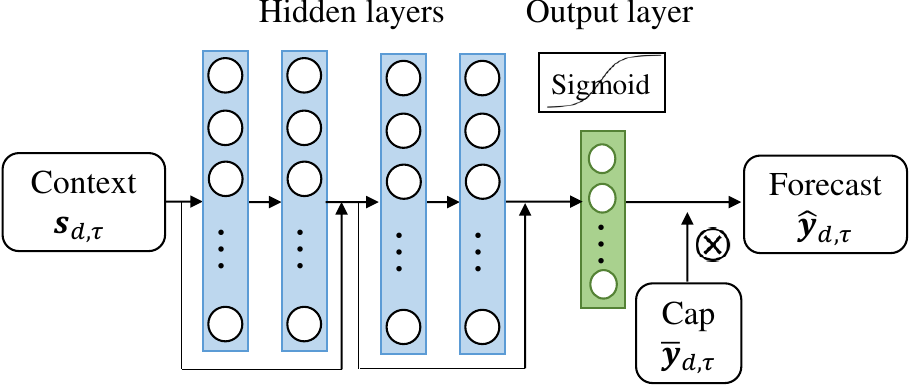}\\
  \caption{The structure of neural network.}
  \label{MLP}
\end{figure}

During the training, the NN, with the parameter $\Theta$ given by any value, outputs the forecast $\hat{\bm{y}}_{d,\tau}$ by \eqref{fm}. The primal and dual problems of DA and RT clearing in \eqref{DAcompact},\eqref{RTcompact},\eqref{RTcompact2} are solved. The active constraint index sets $\mathcal{J}^a_{\text{DA},d},\mathcal{J}^a_{\text{RT},d,\tau-1},\forall \tau=2,...,T$ and the dual solutions $\bm{\sigma}_d^*,\bm{\nu}_{d,1}^*,\bm{\zeta}_{d,\tau}^*,\forall \tau=2,...,T$ are obtained. We check if the active constraint index sets are in the buffers $\Lambda_{\text{DA}}\Lambda_{\text{RT},\tau-1},\forall \tau=2,...,T$. If yes, the stored derivatives $\frac{\partial f_{d,\tau}^x(\hat{\bm{y}}_d)}{\partial \hat{\bm{y}}_d}$,  $\frac{\partial f_{d,\tau-1}^p(\hat{\bm{y}}_d)}{\partial \hat{\bm{y}}_d}$, $\frac{\partial f_{d,\tau-1}^{+-}(\hat{\bm{y}}_d)}{\partial \hat{\bm{y}}_d}$ in the buffers are used for calculating the gradient $ \frac{\partial\ell_d(\hat{\bm{y}}_d)}{\partial\hat{\bm{y}}_d}$ in \eqref{gradient}. If not, we calculate the derivatives associated with the new ones, and calculate the gradient in \eqref{gradient}. Additionally, we store the new active index sets and the associated derivatives in the buffers. The illustration of gradient calculation is in Fig. \ref{code}.

With the gradient $ \frac{\partial\ell_d(\hat{\bm{y}}_d)}{\partial\hat{\bm{y}}_d}$, we use the gradient descent for updating $\Theta$ in \eqref{training_nn}. Recall that $\Pi_\mathcal{J}$ denotes the operator that extracts the elements of a vector corresponding to the indices in the index set $\mathcal{J}$. The subgradient corresponds to the forecast $\hat{\bm{y}}_{d,\tau}$ in $\frac{\partial \ell_d(\hat{\bm{y}}_d)}{\partial \hat{\bm{y}}_d}$ is obtained via $\Pi_{\mathcal{J}_{\tau}}(\frac{\partial \ell_d(\hat{\bm{y}}_d)}{\partial \hat{\bm{y}}_d}) \in \mathbb{R}^N$. The gradient descent for updating the parameter $\Theta$ is,
\begin{equation}\label{gd}
    \Theta \leftarrow \Theta-\alpha \frac{1}{B\cdot T}\sum_{d=1}^B\sum_{\tau=1}^T\Pi_{\mathcal{J}_{\tau}}(\frac{\partial \ell_d(\hat{\bm{y}}_d)}{\partial \hat{\bm{y}}_d})^\top\bigtriangledown_{\Theta}g(\bm{s}_{d,\tau};\Theta)
\end{equation}
where $\alpha$ is the learning rate.

\begin{figure*}[t]
  \centering
  \includegraphics[scale=0.5]{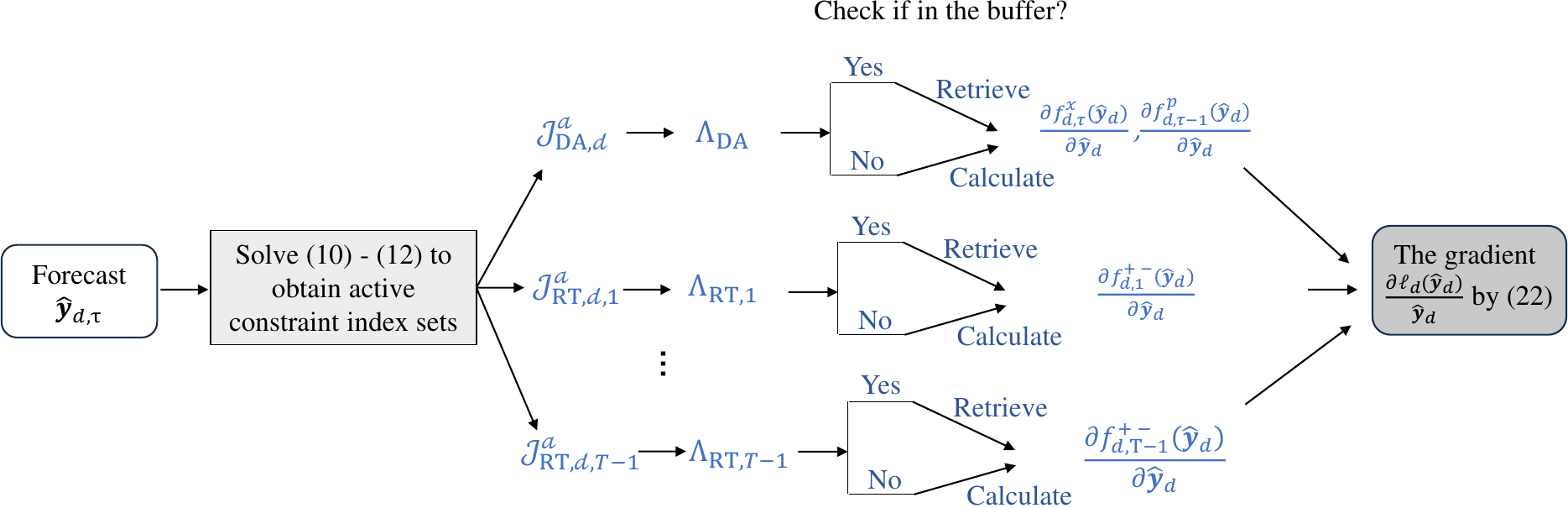}\\
  \caption{The illustration of gradient calculation in the training process.}\label{code}
  % \caption{The flowchart of the algorithm for the training phase. Require: Learning rate $\alpha$, batch size $B$, initialized buffers $\Lambda_{\text{DA}},\Lambda_{\text{RT},\tau},\forall \tau=1,...,T$ for storing active constraint index set/derivative, and initialized forecasting model parameters $\Theta^{(0)}$.}
  \vspace{-1em}
\end{figure*}

\section{Case Study}
We consider a modified version of IEEE 9-bus system.
% ~\cite{chow1982time}. 
As shown in Fig. \ref{9bus}, the system consists of 3 loads, 2 wind farms, and 3 generators ($G_1,G_2,G_3$) whose generation needs to be settled in the DA market and can be adjusted for providing up- and down-regulation power in the RT market. 
% The 3 loads are connected to the nodes \#5, \#7, \#9, respectively, and therefore denoted as $l_5,l_7,l_9$. 
The  generators submit the marginal generation cost $\bm{\rho}$, the minimum generation power $\underline{\bm{p}}$, the maximum generation power $\overline{\bm{p}}$, the ramping limits $\underline{\bm{r}},\overline{\bm{r}}$
in the DA market, which are provided in Table \ref{DA market}. The marginal up-regulation cost $\bm{\rho}_+$, the marginal down-regulation utility $\bm{\rho}_-$, the marginal opportunity loss $\bm{\rho}_+-\bm{\rho},\bm{\rho}-\bm{\rho}_-$ and the adjustment limits $\overline{\bm{p}^+}$ and $\overline{\bm{p}^-}$, that generators submit in the RT market, are provided in Table \ref{DA market} as well. The yearly demand consumption data is used, with a valley value of 210 \unit{MW}, and a peak value of 265 \unit{MW}. The hourly wind power production in the year of 2012 from GEFCom 2014 is used, whose range is from 0 to 1. The dataset consists of the hourly numerical weather prediction (the predicted wind speed and direction at altitudes of 10 meters and 100 meters) and hourly wind power outputs. 80\% data is divided into the training set, while the rest forms the test set. In this way, the training set has 7008 samples, while the test set has 1752 samples. The wind data is scaled by multiplying a constant according to the considered wind generation capacity, which will be discussed in the following. The demand and wind data are in \cite{ppm2022}.

\begin{figure}[t]
  \centering
  \includegraphics[scale=0.4]{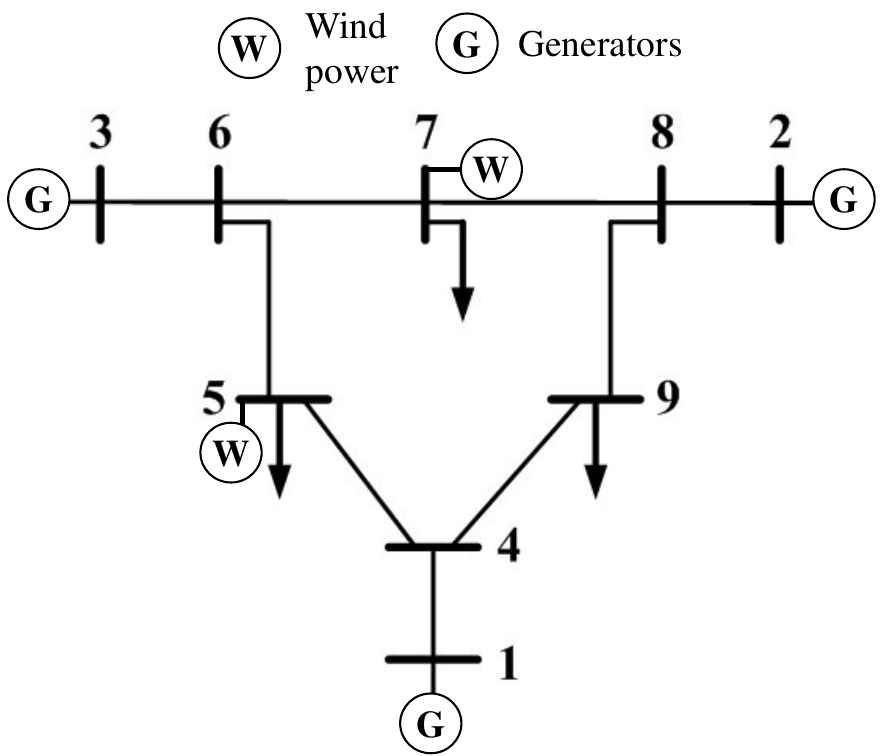}\\
  \caption{The illustration of IEEE 9-bus system.}\label{9bus}
  \vspace{-1em}
\end{figure}

\begin{table}[ht]
\caption{Cost and technical data in DA and RT market of IEEE 9-Bus system.}\label{DA market}
\begin{center}
\begin{tabular}{ c c  c  c  c}
\hline\hline
& & $G_1$ & $G_2$ & $G_3$ \\
\hline
         &\makecell[c]{Marginal generation cost $\bm{\rho}$ (\$/MW)} & 20 & 22 & 24\\
         &\makecell[c]{Minimum generation power $\underline{\bm{p}}$ (MW)} & 0 & 0 & 0\\
         \makecell[c]{DA\\ market}&\makecell[c]{Maximum generation power $\overline{\bm{p}}$ (MW)} & 150 & 200 & 270\\
         &\makecell[c]{Lower ramping limits $\underline{\bm{r}}$ (MW)} & -90 & -80 & -70\\
         &\makecell[c]{Upper ramping limits $\overline{\bm{r}}$ (MW)} & 90 & 80 & 70\\
\hline
&\makecell[c]{Marginal  cost $\bm{\rho}_+$ (\$/MW)} & 50 & 52 & 54\\
&\makecell[c]{Marginal  utility $\bm{\rho}_-$ (\$/MW)} & 18 & 16 & 14\\
&\makecell[c]{Marginal opportunity loss \\for up-regulation $\bm{\rho}_+-\bm{\rho}$ (\$/MW)} & 30 & 30 & 30\\
\makecell[c]{RT \\Market}&\makecell[c]{Marginal opportunity loss \\for down-regulation $\bm{\rho}-\bm{\rho}_-$ (\$/MW)} & 2 & 6 & 10\\
&\makecell[c]{Up-regulation limit $\overline{\bm{p}^+}$ (MW)} & 60 & 60 & 60\\
&\makecell[c]{Down-regulation limit $\overline{\bm{p}^-}$ (MW)} & 60 & 60 & 60\\
\hline\hline
\end{tabular}
\end{center}
\vspace{-2em}
\end{table}

We use a four-layer ResNet as the forecasting model, which has 256 hidden layer units. Its structure is described in Figure~\ref{MLP}. The input context consists of the numeric weather prediction of each wind farm in the system.  We select the type of forecasting model and its associated hyperparameters based on experiences \cite{zhang2024valueoriented}. Based on our past work, a multilayer perceptron (MLP) with two hidden layers and 256 hidden units performs well for value-oriented forecasting in simple decision-making problems. For more complex decision-making problems considered in this work, we use ResNet shown in Fig. \ref{MLP} as the forecasting model, which is composed of two MLPs but with skip layers. This allows the model to have a larger capacity when needed; otherwise, it reduces to the MLP used in \cite{zhang2024valueoriented}. 
We use Root Mean Squared Error (RMSE) on the test set for assessing the \emph{forecast quality}, and the average overall operating cost on the test set, as defined in \eqref{overallcost}, for evaluating the \emph{operation value}. Four benchmark models are used for comparison: Two quality-oriented (\textbf{Qua-E} and \textbf{Qua-Q}), one value-oriented point forecasting approach (\textbf{OptNet}), and a stochastic market clearing (\textbf{Sto-OPT-P}), i.e., 

\begin{enumerate}
\item \textbf{Qua-E}: The forecasting model (ResNet) is trained with the loss function of Mean Squared Error (MSE). The trained model predicts conditional expectations.
\item \textbf{Qua-Q}: The forecasting model is trained with the pinball loss. The trained model offers quantile predictions. Light Gradient Boosted Machine, the winner of GEFCom 2014 \cite{landry2016probabilistic}, is used as the forecasting model.
\item \textbf{OptNet}: The forecasting model (ResNet) is trained by differentiable programming \cite{wahdany2023more,donti2017task}, to minimize overall system operating costs.
\item \textbf{Sto-OPT-P}: For stochastic clearing, 50 wind power scenarios are obtained via k-nearest-neighbors, which is a Prescriptive approach  \cite{bertsimas2020predictive}.  For each sample on the test set, the stochastic clearing is for settling the schedule of generators and wind power in DA. Then, the adjustment in \eqref{RTcompact} and \eqref{RTcompact2} are performed in RT.
\end{enumerate}

In the following, we report the results on the test set. The results on the training set are provided in the Appendix \ref{training results} to assess potential overfitting of the forecasting models and to illustrate the training dynamics.
\vspace{-1em}
\subsection{The Operational Advantage}
The capacity of two wind farms is set as 105 \unit{MW}, respectively, which takes up 79\% of the maximum demand. The nominal level of Qua-Q issued quantile is chosen as $\frac{1}{16}$. Such nominal level is determined as the one beneficial to maximizing wind power profit, i.e., $\frac{\bm{\rho}^1-\bm{\rho}_{-}^1}{\bm{\rho}_+^1-\bm{\rho}_-^1}$ \cite{pinson2007trading}, where $\bm{\rho}^1,\bm{\rho}_+^1,\bm{\rho}_-^1$ represent the DA marginal cost of $G_1$, as well as the RT marginal up-regulation cost and down-regulation utility. The results of RMSE and average operating cost, along with training time per epoch and test time, are in Table \ref{tab3}. 

Since Sto-OPT-P does not need training or relies on a point forecast, its RMSE is not reported. Sto-OPT-P serves as the ideal benchmark \cite{morales2014electricity}, which has the least average operating cost. The proposed approach outperforms all other methods in terms of average operating cost on the test set, except for Sto-OPT-P.  However, its test time is much shorter than Sto-OPT-P, demonstrating computational efficiency. Also, 
we observe that the performance of Sto-OPT-P is heavily influenced by the number of scenarios used. When fewer scenarios, such as 20, are employed, the average operating cost on the test set increases to \$84,478, which is even worse than that achieved by the proposed approach.

The proposed approach exhibits a higher RMSE compared to Qua-E. This underscores the point that \emph{more accurate forecasts don't always translate to better operational performance}. Additionally, since the incremental bidding price $\bm{\rho}_+-\bm{\rho}$ is larger than $\bm{\rho}-\bm{\rho}_-$, the marginal opportunity loss of up-regulation is larger than that of down-regulation. Therefore, Qua-Q, which issues the quantile forecasts with a low proportion level ($\frac{1}{16}$), has better performance than Qua-E. As for the training time, the proposed approach requires a longer training time than Qua-E due to the more complex computation involved in calculating the gradient during the training process. But it is still acceptable, and much shorter than the value-oriented forecasting approach OptNet.

\begin{table}[ht]
\caption{RMSE, average operating cost on the test set and the training/test time of the proposed approach and benchmarks.}\label{tab3}
\begin{center}
\begin{tabular}{ c  c  c  c c c}
\hline\hline
         &\makecell[c]{Proposed} & \makecell[c]{Qua-E} & \makecell[c]{Qua-Q} & \makecell[c]{OptNet} & \makecell[c]{Sto-OPT-P}\\
\hline
    RMSE (\unit{MW}) & 26 & 18 & 29 & 33  & -\\
    \makecell[c]{Average operating\\ cost (\unit{\$})} & 84449 & 86990 & 85154 & 86347 & 84362\\
    \makecell[c]{Training time\\ per epoch} & 25 s & 0.08 s & - & 136 s & -\\
    \makecell[c]{Test time} & 5 s & 5 s & 5 s & 5 s & 64min\\
\hline\hline
\end{tabular}
\end{center}
\vspace{-2em}
\end{table}

Apart from reducing the DA and the RT operating costs, the wind farm profit can also be benefited from the proposed approach. Table \ref{windprofits} summarizes the average profits of wind farms obtained under the proposed approach and Qua-E. The results indicate that wind farms can achieve higher profits with the proposed method compared to Qua-E. The intuition behind this result is as follows: the proposed method is trained to minimize the overall costs of the DA and RT markets, which is equivalent to maximizing the overall social welfare of these markets. Since wind farm profits are a component of overall social welfare, maximizing social welfare implicitly enhances wind farm profits to some extent. In contrast, Qua-E fully ignores the impact of forecasts on wind power profits, leading to its inferior performance.

\begin{table}[h]
\caption{Average profits of wind farms in the test set of the proposed approach and Qua-E.}\label{windprofits}
\begin{center}
\begin{tabular}{ c  c  c  c}
\hline\hline
      & 
     \makecell[c]{Total profits of \\
     wind farms 5 and 7/\$} & \makecell[c]{Profit of \\ wind farm 5/\$}& \makecell[c]{Profit of \\ wind farm 7/\$}\\
\hline
    \makecell[c]{Proposed}  &
    24708  &11140&13568\\
\hline
    \makecell[c]{Qua-E}  &
    17876  &5491&12386\\
\hline\hline
\end{tabular}
\end{center}
\end{table}

Additionally, we present the distribution on the test set of the DA generation schedule of traditional generators, along with their RT up- and down- adjustment in Fig. \ref{powerdis}. The proposed method tends to forecast less power to avoid the costly up-regulation, compared to Qua-E. Therefore, in the DA market, traditional generators produce more power under the proposed method due to the lower renewable energy forecasts. Accordingly, in the RT market, the proposed forecast leads to lower up-regulation needs and higher down-regulation needs compared to Qua-E.

\begin{figure}[t]
  \centering
  \includegraphics[scale=0.6]{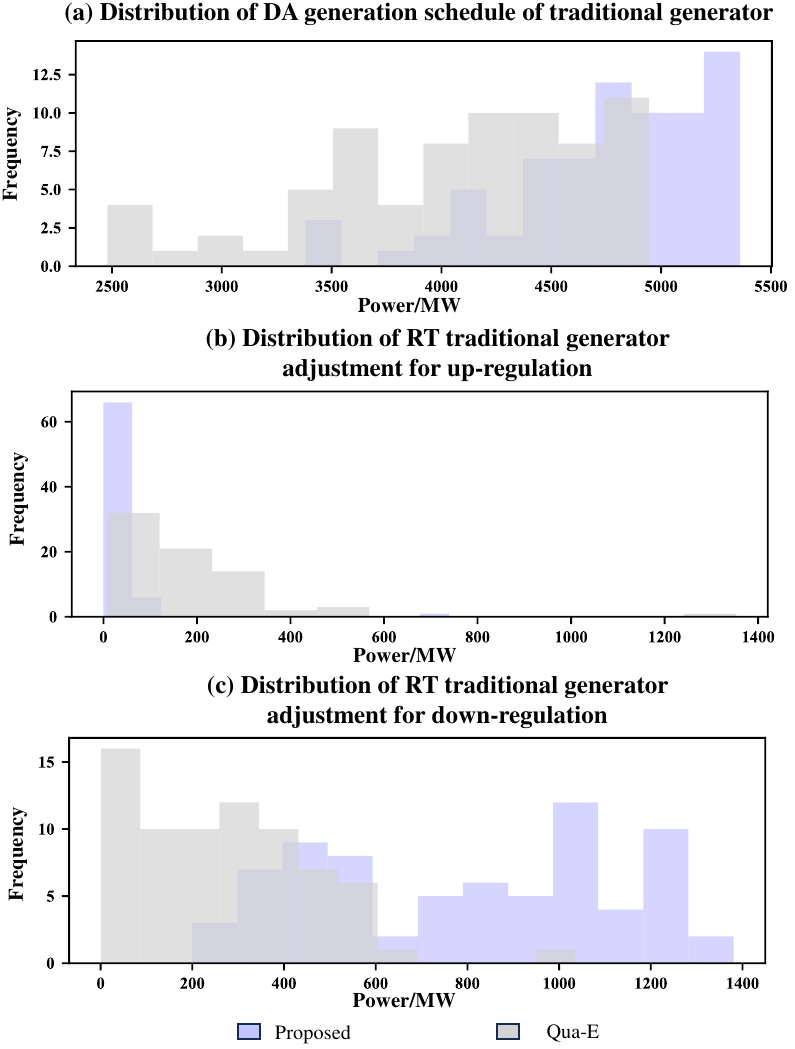}
  \caption{Distribution of hourly generator day-ahead schedules and real-time adjustments over 73 test days.}
\label{powerdis}
\end{figure}

\subsection{The Sensitivity Analysis}
In this section, we compare the proposed approach against Qua-E under different wind power capacities and marginal cost $\bm{\rho}^+$ for up-regulation.

% We compare the operational advantage in terms of reduced sum of DA and RT costs in two settings. Firstly, under the cost setting in Tables 1 and 2, we run the experiments under different wind power capacity. Secondly, we  maintain the wind power capacity at 105 \unit{MW}, and test the approach under different per unit regulation cost/benefit in the RT market. The results obtained by our approach in Algorithms \ref{alg2} and 2 are reported, and denoted as alg1/alg2, respectively.

\subsubsection{Performance under Different Wind Power Penetration}

Here, different wind power capacities are considered, i.e., 85 \unit{MW}, 95 \unit{MW}, and 105 \unit{MW} per wind farm. Fig. \ref{Capacities}(a) shows the average operating cost of the proposed approach and Qua-E under different wind power capacities.  Under large wind power capacity, the average cost reduction of the proposed approach is more obvious. For instance, such a reduction is 2.4\% and 2.9\%, respectively, under the wind power capacity of 85 \unit{MW} and 105 \unit{MW}.  The proposed approach has larger operation benefits under larger penetration of wind power. 

\begin{figure}[h]
  \centering
  \includegraphics[scale=0.55]{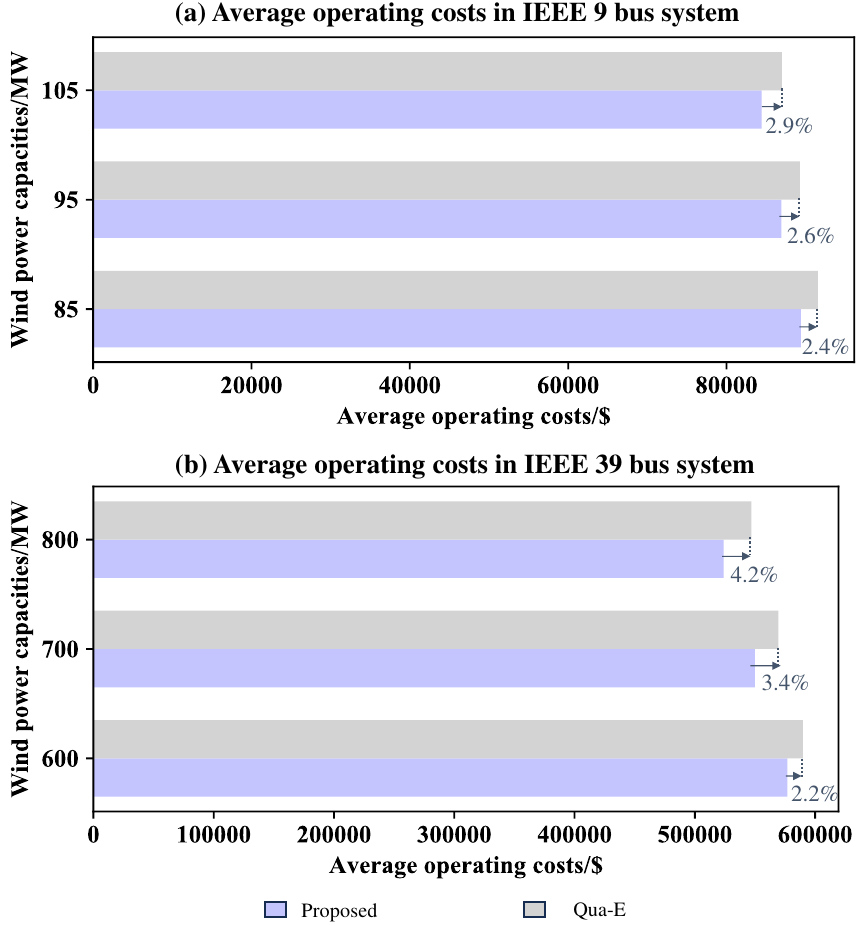}\\
  \caption{Average operating cost under different wind power capacities. Qua-E stands for quality-oriented forecast issuing expectation. }
  \label{Capacities}
\end{figure}

\subsubsection{Performance under Different Up-regulation Cost in RT}
The performance is further tested under various up-regulation marginal costs, along with the marginal opportunity loss $\bm{\rho}_+-\bm{\rho}$ for up-regulation. The marginal opportunity loss $\bm{\rho}-\bm{\rho}_-$ for down-regulation is the same as in Table \ref{DA market}. The capacity of wind power is set to 105 \unit{\MW}. The average operating cost of the proposed approach and Qua-E under two settings are listed in Table \ref{AWS setting}. When the RT market lacks flexibility for up-regulation, the up-regulation marginal cost is high, and the marginal opportunity loss for up-regulation is much larger than that for down-regulation. Therefore, the proposed approach tends to forecast less power than Qua-E to mitigate the risk of costly up-regulation, and results in a significant cost reduction (8\%). When the up-regulation marginal cost is similar to the DA marginal generation cost $\bm{\rho}$, the marginal opportunity loss for up-regulation is lower than that for down-regulation. Therefore, the proposed approach tends to forecast more power to mitigate the risk of costly down-regulation. In such a case, since the marginal opportunity losses of up- and down-regulation are very similar, the cost reduction of the proposed approach is less significant. The forecast profiles for 6 days of wind farm at the node 5 under the two settings are given in Fig. \ref{profile}.

To sum up, the operation advantage of the proposed approach is more evident, under large penetration of wind power, and high up-regulation marginal cost.
% \todo{Name setting 1, 2, 3 to be more informative, like high, normal, low up-regulation cost in Table V; also in the text}

% \begin{table}[h]
% \caption{Different settings of up-regulation costs and marginal opportunity loss in RT, where the up-regulation costs are given on the left of slash, and the marginal opportunity loss are given on the right.}\label{setting}
% \begin{center}
% \begin{tabular}{ c  c  c  c}
% \hline\hline
%          & $G_1$ & $G_2$ & $G_3$\\
% \hline
%     \makecell[c]{High up-regulation\\ cost setting} & 80/60 & 82/60 & 84/60\\
% \hline
%     \makecell[c]{Low up-regulation\\ cost setting} & 21/1 & 23/1 & 25/1\\
% \hline\hline
% \end{tabular}
% \end{center}
% \end{table}

\begin{table}[h]
\caption{Average operating cost in different settings of up-regulation costs in RT. The first column gives the RT up-regulation costs and marginal opportunity loss in each setting. Inside the parenthesis, the up-regulation marginal costs are given on the left of the slash, and the marginal opportunity loss is given on the right.}\label{AWS setting}
\begin{center}
\begin{tabular}{ c  c  c  c}
\hline\hline
         Settings& \makecell[c]{Proposed (\$)} & \makecell[c]{Qua-E (\$)}&\makecell[c]{Cost\\ reduction 
 (\%)} \\
\hline
%     \makecell[c]{Normal up-\\regulation cost} & \makecell[c]{4530} & 4671 & \makecell[c]{3\%}\\
% \hline
    \makecell[c]{High up-\\regulation marginal cost \\ ($G_1$: 80/60, $G_2$:\\ 82/60, $G_3$: 84/60)} & \makecell[c]{85114} & 92486 & \makecell[c]{8\%}\\
\hline
    \makecell[c]{Low up-\\regulation marginal cost\\ ($G_1$: 21/1, $G_2$:\\ 23/1, $G_3$: 25/1)} & \makecell[c]{81517} & 81677 & \makecell[c]{0.2\%}\\
\hline\hline
\end{tabular}
\end{center}
\end{table}

\vspace{-1em}
\begin{figure}[ht]
  \centering
  \includegraphics[scale=0.55]{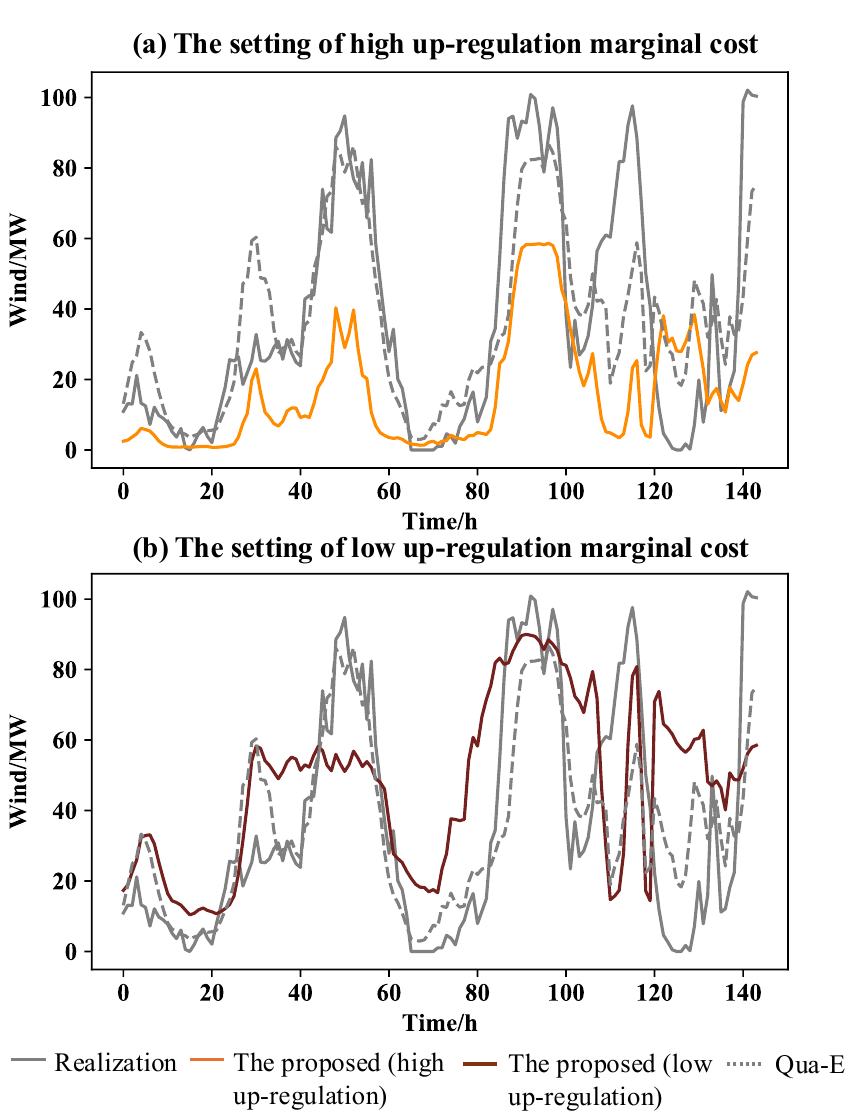}\\
  \caption{The 6-day wind power forecast profiles for a wind farm connected to Node 5 are provided by Qua-E and the proposed approach under both high and low up-regulation cost settings.}\label{profile}
\end{figure}

\subsection{Extending to Larger Scale System}

In the following, we test the proposed approach on the IEEE 39-bus system and a 500-bus system developed in ARPA-e's Grid
Optimization Competition \cite{birchfield2016grid}. The IEEE 39-bus system has 21 loads, 2 wind farms at buses 38 and 39, and 10 generators.  The illustration of the IEEE 39-bus system and parameters for generators are given in our Github repo~\cite{ppm2022}. The average operating costs of the proposed one and Qua-E are shown in Fig. \ref{Capacities}(b). The proposed approach results in lower average operating costs than Qua-E. With increased wind power capacity in the IEEE 39-bus system, the advantage of the proposed approach in decreasing the overall costs is even more pronounced compared to the IEEE 9-bus system.

The 500-bus system consists of 281 loads, 224 generators, and two wind farms located at buses 45 and 46, each with a capacity of 2000 \unit{MW}. Detailed parameters are available in the corresponding GitHub repository \cite{bus500}. We evaluate the proposed approach under two transmission capacity scenarios. In the no-congestion case, transmission line capacities are sufficiently large to prevent congestion. In the congestion case, the capacities are reduced, resulting in congestion on certain lines. We compare the proposed approach with Qua-E under two cases. The results are shown in Table \ref{500bus}.  

The proposed approach achieves lower average operating costs compared to Qua-E. Additionally, for both methods, the operating cost under the congestion scenario is higher than that under the no-congestion scenario. This outcome is expected, as congestion limits system flexibility and leads to higher costs. We also note that training is performed offline, and the computational cost is particularly high for large-scale systems. For example, on the 500-bus system, training for one epoch with 21 batches takes approximately 3 minutes on a 2024 MacBook Pro. The computations required for gradient calculation, involving the clearing of day-ahead and real-time markets, as well as the matrix inversion, are highly demanding for large-scale systems. Future research is needed to address such computational challenges. We also note that, although offline training is computationally expensive, issuing forecasts during the operational phase is as computationally efficient as common forecasting approaches such as Qua-E, since it only involves the forward propagation of a neural network.

\begin{table}[h]
\caption{Average operating cost in different cases of transmission line capacities.}\label{500bus}
\begin{center}
\begin{tabular}{ c  c  c  c}
\hline\hline
         Settings& \makecell[c]{Proposed (\$)} & \makecell[c]{Qua-E (\$)}&\makecell[c]{Cost\\ reduction 
 (\%)} \\
\hline
%     \makecell[c]{Normal up-\\regulation cost} & \makecell[c]{4530} & 4671 & \makecell[c]{3\%}\\
% \hline
    \makecell[c]{No-congestion} & \makecell[c]{1631631} & 1648846 & \makecell[c]{1\%}\\
\hline
    \makecell[c]{Congestion} & \makecell[c]{1635312} & 1649040 & \makecell[c]{0.8\%}\\
\hline\hline
\end{tabular}
\end{center}
\end{table}

% \begin{figure}[h]
%   \centering
%   % Requires \usepackage{graphicx}
%   \includegraphics[scale=0.6]{Figures/bus39.pdf}\\
%   \caption{Average operating cost under different wind power capacities in IEEE 39-Bus system.}
%   \label{bus39}
% \end{figure}

%\vspace{-1em}
\section{Conclusions}
% \todo{mention networked case} \todo{be more concise, combine these two paragraphs as one paragraph}
We propose a value-oriented renewable energy forecasting approach, for minimizing the expected overall operating cost in the existing deterministic market clearing framework. We analytically derive the loss function for value-oriented renewable energy forecasting in sequential market clearing. The loss function is proved to be piecewise linear when the market clearing is modeled by linear programs. Additionally, we provide the analytical gradient of the loss function with respect to the forecast, which leads to an efficient training strategy. In the case study, compared to quality-oriented forecasting approach trained by MSE, the proposed approach can reduce average operating cost on the test set to 2.9\% for the IEEE 9-bus system. Such an advantage is more obvious under large wind power capacity and high up-regulation costs. Under high up-regulation marginal costs, our approach can reduce the cost by up to 8\%. 

Future work will include derivation of the value-oriented loss function for market clearing modeled by other types of optimization programs, such as quadratic optimization and conic optimization. Additionally, the development of value-oriented RES forecasting depends on the structure of the decision-making problem. If this structure changes, the forecasting model needs to be redeveloped in the current offline training approach. Addressing how to adapt such forecasts to changes in decision-making will be a focus of future work. Additionally, we note that the effectiveness of value-oriented forecasts relative to quality-oriented forecasts is case-dependent, influenced by specific decision-making structures and parameters. In certain scenarios, value-oriented forecasts may converge to quality-oriented ones, resulting in comparable performance. We would like to leave the theoretical analysis of the improvement bounds to future work.

% Furthermore, this work focuses on providing a theoretical approach for deriving the loss function used to train value-oriented forecasting models. The structure of the loss function offers insights into adapting this approach for large-scale systems in the future.
% \revise{Also, the proposed approach will be extended to consider both wind and load uncertainty.} 
% \todo{mention 1-2 more like the demand uncertainty mentioned early} 

\vspace{-1em}

\bibliographystyle{IEEEtran}
\bibliography{IEEEabrv,mylib}

\appendices
\section{Compatibility of the Market Clearing Model with Energy Storage}\label{compat}

We follow the practice in \cite{luth2024electrolysis} using the linear model for energy storage. We aim to show that the structure of market clearing remains the same as those in \eqref{DAcompact}, \eqref{RTcompact}, and \eqref{RTcompact2}, when incorporating energy storage. In the DA market clearing, let $\bm{L}_{d,\tau},\bm{G}_{d,\tau}$ denote the charge and discharge power of energy storage at time $\tau$ on day $d$, and $\bm{S}_{d,\tau}$ denote the stored energy. The model of energy storage in DA market is, 
\begin{subequations}\label{DAES_m}
\begin{alignat}{2}
& \mathop{\min}_{\bm{L}_{d,\tau},\bm{G}_{d,\tau},\bm{S}_{d,\tau}}   &&\sum_{\tau=1}^T\bm{\rho}^{G\top}\bm{G}_{d,\tau}-\bm{\rho}^{L\top}\bm{L}_{d,\tau}\label{DAobja}\\ 
    & \text{s.t.} && 0 \leq \bm{L}_{d,\tau} \leq \overline{\bm{L}}, \forall \tau=1,...,T\label{DAconstra}\\
&&& 0 \leq \bm{G}_{d,\tau} \leq \overline{\bm{G}}, \forall \tau=1,...,T\label{DAconstrb}\\
&&& \bm{S}_{d,\tau}=\bm{S}_{d,\tau-1}+\eta\bm{L}_{d,\tau}-\bm{G}_{d,\tau}/\eta,\nonumber\\
&&&\forall \tau=2,...,T\label{DAconstrc}\\
&&& \bm{S}_{d,1}=\bm{S}_0+\eta\bm{L}_{d,1}-\bm{G}_{d,1}/\eta,\forall \tau=1\label{DAconstrd}\\
&&& \underline{\bm{S}} \leq \bm{S}_{d,\tau} \leq \overline{\bm{S}}, \forall \tau=1,...,T\label{DAconstre},
\end{alignat} 
\end{subequations}
where $\bm{\rho}^G$ and $\bm{\rho}^L$ are the marginal cost and utility for charge and discharge, and $\eta$ is the charging efficiency. \eqref{DAconstra} and \eqref{DAconstrb} limit the charge and discharge power within their respective upper bounds. \eqref{DAconstrc} and \eqref{DAconstrd} determine the stored energy based on the charged and discharged power, while \eqref{DAconstre} constrains it within the lower and upper bounds. After the DA market clearing, the optimal solutions are denoted as $\bm{L}_{d,\tau}^*,\bm{G}_{d,\tau}^*,\bm{S}_{d,\tau}^*$.

In the RT market, energy storage operates similarly to conventional electricity generation technologies in their discharge mode and similarly to the demand in the charge mode. The adjustments for up- and down- regulation in the discharge mode are denoted as $\bm{G}_{d,\tau}^+,\bm{G}_{d,\tau}^-$, respectively. The adjustments for up- and down- regulation in the charge mode are denoted as $\bm{L}_{d,\tau}^+,\bm{L}_{d,\tau}^-$, respectively. The model of the energy storage in the RT market at time $\tau=1,...,T$ is,
\begin{subequations}\label{RTES_m}
\begin{alignat}{2}
& \mathop{\min}_{\bm{L}_{d,\tau}^+,\bm{G}_{d,\tau}^+,\bm{L}_{d,\tau}^-,\bm{G}_{d,\tau}^-}   &&\bm{\rho}^{G\top}_+\bm{G}_{d,\tau}^++\bm{\rho}^{L\top}_+\bm{L}_{d,\tau}^+-\nonumber\\
&&& \qquad(\bm{\rho}^{G\top}_-\bm{G}_{d,\tau}^-+\bm{\rho}^{L\top}_-\bm{L}_{d,\tau}^-)\label{RTESobja_m}\\ 
    & \text{s.t.} && 0 \leq \bm{L}^*_{d,\tau}-\bm{L}^+_{d,\tau}+\bm{L}^-_{d,\tau} \leq \overline{\bm{L}}\label{RTconstra_m}\\
&&& 0 \leq \bm{G}^*_{d,\tau}+\bm{G}^+_{d,\tau}-\bm{G}^-_{d,\tau} \leq \overline{\bm{G}}\label{RTconstrb_m}\\
&&& \underline{\bm{S}} \leq \bm{S}^*_{d,\tau}+\eta(-\bm{L}^+_{d,\tau}+\bm{L}^-_{d,\tau})-\nonumber\\
&&&\qquad (\bm{G}^+_{d,\tau}-\bm{G}^-_{d,\tau})/\eta \leq \overline{\bm{S}}\label{RTconstrc_m}.
\end{alignat} 
\end{subequations}
where $\bm{\rho}_+^G,\bm{\rho}_+^L$ are the marginal cost for up-regulation in the discharge and charge modes, respectively. $\bm{\rho}_-^G,\bm{\rho}_-^L$ are the marginal utility for down-regulation in the discharge and charge modes, respectively.
The constraints in \eqref{RTES_m} ensure that, after the adjustment, the charge, discharge, and stored energy constraints remain satisfied. With the energy storage models defined in \eqref{DAES_m} and \eqref{RTES_m}, we are ready to present the formulations of the DA and RT market clearing problems with energy storage. The DA market clearing is,
\begin{subequations}\label{DA_ESS}
\begin{alignat}{2}
&\mathop{\min}_{\bm{x}_{d}}&&\ \bm{\rho}^\top\bm{p}_{d,\tau}+\bm{\rho}^{G\top}\bm{G}_{d,\tau}-\bm{\rho}^{L\top}\bm{L}_{d,\tau}\\
& \text{s.t.} &&   \bm{1}^\top(\bm{p}_{d,\tau}+\bm{w}_{d,\tau}-\bm{L}_{d,\tau}+\bm{G}_{d,\tau})=\bm{1}^\top\bm{l}_{d,\tau},\nonumber\\
&&&\forall \tau=1,...,T
    \\ 
    &&&  
    -\overline{\bm{f}}\leq \bm{H}(\bm{p}_{d,\tau}+\bm{w}_{d,\tau}-\bm{l}_{d,\tau}-\bm{L}_{d,\tau}+\bm{G}_{d,\tau})\leq\overline{\bm{f}},\nonumber\\
    &&&\forall \tau=1,...,T\\
    &&& \eqref{DAd},\eqref{DAe},\eqref{DAf},\eqref{DAconstra},\eqref{DAconstrb},\eqref{DAconstrc},\eqref{DAconstrd},\eqref{DAconstre} \nonumber
\end{alignat}
\end{subequations}

The optimal solution of \eqref{DA_ESS} is denoted as $\bm{x}_d^*=[\bm{x}_{d,\tau}^*]_{\tau=1}^T=[\bm{p}_{d,\tau}^*;\bm{w}_{d,\tau}^*;\bm{L}_{d,\tau}^*;\bm{G}_{d,\tau}^*;\bm{S}_{d,\tau}^*]_{\tau=1}^T$, which includes the generation schedule $\bm{p}_{d,\tau}^*,\bm{w}_{d,\tau}^*$ of flexible generators and RES, as well as the DA schedule $\bm{L}_{d,\tau}^*,\bm{G}_{d,\tau}^*,\bm{S}_{d,\tau}^*$ of energy storage.

The RT market clearing with energy storage at time $\tau=1$ is,
\begin{subequations}\label{RTreform_ES}
\begin{alignat}{2}
&\mathop{\min}_{\bm{z}_{d,\tau}}&&\ \bm{\rho}_+^\top\bm{p}_{d,\tau}^+-\bm{\rho}_-^\top\bm{p}_{d,\tau}^-+\bm{\rho}^{G\top}_+\bm{G}_{d,\tau}^++\bm{\rho}^{L\top}_+\bm{L}_{d,\tau}^+-\nonumber\\
&&& \ (\bm{\rho}^{G\top}_-\bm{G}_{d,\tau}^-+\bm{\rho}^{L\top}_-\bm{L}_{d,\tau}^-)\\
& \text{s.t.} && \bm{1}^\top(\bm{p}_{d,\tau}^+-\bm{p}^-_{d,\tau}-\bm{\kappa}_{d,\tau}+\bm{G}^+_{d,\tau}-\bm{G}^-_{d,\tau}\nonumber\\
&&&+\bm{L}^+_{d,\tau}-\bm{L}^-_{d,\tau})=-\bm{1}^\top(\bm{y}_{d,\tau}-\bm{w}^*_{d,\tau})\label{RTaES}
    \\ 
    &&& 
    -\overline{\bm{f}}-\bm{H}(\bm{p}_{d,\tau}^*+\bm{w}^*_{d,\tau}-\bm{l}_{d,\tau})\leq \bm{H}(\bm{p}_{d,\tau}^+-\bm{p}_{d,\tau}^--\bm{\kappa}_{d,\tau}\nonumber\\
    &&& \qquad +\bm{G}^+_{d,\tau}-\bm{G}^-_{d,\tau}+\bm{L}^+_{d,\tau}-\bm{L}^-_{d,\tau}+\bm{y}_{d,\tau}-\bm{w}^*_{d,\tau})\nonumber\\
    &&&\leq\overline{\bm{f}}-\bm{H}(\bm{p}_{d,\tau}^*+\bm{w}^*_{d,\tau}-\bm{l}_{d,\tau})\label{RTbES}\\
&&& \eqref{RTd},\eqref{RTe},\eqref{reformuRT1simp},\eqref{RTg},\eqref{RTconstra_m},\eqref{RTconstrb_m},\eqref{RTconstrc_m}
\end{alignat}
\end{subequations}

The RT market clearing with energy storage at time $\tau=2,...,T$ becomes,
\begin{subequations}\label{RTreform2ES}
\begin{alignat}{2}
&\mathop{\min}_{\bm{z}_{d,\tau}}&&\ \bm{\rho}_+^\top\bm{p}_{d,\tau}^+-\bm{\rho}_-^\top\bm{p}_{d,\tau}^-+\bm{\rho}^{G\top}_+\bm{G}_{d,\tau}^++\bm{\rho}^{L\top}_+\bm{L}_{d,\tau}^+-\nonumber\\
&&& \ (\bm{\rho}^{G\top}_-\bm{G}_{d,\tau}^-+\bm{\rho}^{L\top}_-\bm{L}_{d,\tau}^-)\\
& \text{s.t.} && \eqref{RTaES},\eqref{RTbES},\eqref{RTd},\eqref{RTe},\eqref{reformuRT1simp},\eqref{reformuRT2simp},\eqref{RTg},\eqref{RTconstra_m},\eqref{RTconstrb_m},\eqref{RTconstrc_m}
\end{alignat}
\end{subequations}
where $\bm{z}_{d,\tau}=[\bm{p}_{d,\tau}^+;\bm{p}_{d,\tau}^-;\bm{\kappa}_{d,\tau};\bm{L}_{d,\tau}^+;\bm{L}_{d,\tau}^-;\bm{G}_{d,\tau}^+;\bm{G}_{d,\tau}^-]$ is the collection of RT decision variables.

The compact forms of \eqref{DA_ESS}, \eqref{RTreform_ES}, and \eqref{RTreform2ES} are structurally the same as those in \eqref{DAcompact}, \eqref{RTcompact}, and \eqref{RTcompact2}, except that the constant coefficients differ to accommodate the decision variables associated with energy storage.

\section{Proof of Cost Recovery and Revenue Adequacy}\label{proof_property}

The market clearing \eqref{DA} is equivalent to an equilibrium model that captures the profit maximization problems of traditional generators and RES producers. First, we present the electricity prices, which are determined by the dual variables listed after the colons in \eqref{DAb} and \eqref{DAc}.
\begin{equation}
     \bm{\lambda}_{d,\tau}=\gamma_{d,\tau}\bm{1}+\bm{H}^\top(\underline{\bm{\mu}}_{d,\tau}-\overline{\bm{\mu}}_{d,\tau})
 \end{equation}
 where $\bm{\lambda}_{d,\tau}=[\lambda_{d,\tau,n}]_{n=1}^N$. $\lambda_{d,\tau,n}$ is the electricity price at the node $n$. The profit maximization problem of the traditional generator connected to the node $n$ is,
\begin{subequations}\label{GP}
\begin{alignat}{2}
& \mathop{\max}_{p_{d,\tau,n}}   &&\sum_{\tau=1}^T(\lambda_{d,\tau,n}-\rho_{n})p_{d,\tau,n}\label{Ga}\\ 
    & \text{s.t.} &&  0 \leq p_{d,\tau,n} \leq \overline{p}_n: \underline{\eta}_{d,\tau,n},\overline{\eta}_{d,\tau,n},\forall \tau=1,...,T \label{Gb}\\
    &&& -\overline{r}_n \leq p_{d,\tau,n}-p_{d,\tau-1,n} \leq \overline{r}_n: \underline{\psi}_{d,\tau,n},\overline{\psi}_{d,\tau,n},\nonumber\\
    &&&\forall \tau=2,...,T\label{Gc}
\end{alignat}
\end{subequations}
where $\overline{p}_n,\overline{r}_n$ are the generation and ramping limits, and $\overline{\bm{p}}=[\overline{p}_n]_{n=1}^N,\overline{\bm{r}}=[\overline{r}_n]_{n=1}^N$. $p_{d,\tau,n}$ and $\rho_n$ is the generator schedule and its marginal cost, and $\bm{p}_{d,\tau}=[p_{d,\tau,n}]_{n=1}^N,\bm{\rho}=[\rho_n]_{n=1}^N$ 
The profit maximization problem of RES connected to the node $n$ is,
\begin{subequations}\label{WP}
\begin{alignat}{2}
& \mathop{\max}_{w_{d,\tau,n}}   &&\sum_{\tau=1}^T \lambda_{d,\tau,n}w_{d,\tau,n}\label{Ra}\\ 
    & \text{s.t.} &&  0 \leq w_{d,\tau,n} \leq \hat{y}_{d,\tau,n}: \underline{\delta}_{d,\tau,n},\overline{\delta}_{d,\tau,n},\forall \tau=1,...,T \label{Wb}
\end{alignat}
\end{subequations} 
where $w_{d,\tau,n}$ is the RES schedule with its RIEQ $\hat{y}_{d,\tau,n}$.

Also, the equilibrium model includes the power balance and transmission power limits as market constraints,
\begin{equation}\label{EP}
    \eqref{DAb},\eqref{DAc}
\end{equation}

Based on the equilibrium model given by \eqref{GP}, \eqref{WP}, and \eqref{EP}, we first prove the cost recovery property, followed by revenue adequacy.

\textbf{(a) Proof of cost recovery}

Cost recovery means that profits of market players are nonnegative at the market clearing solutions. Since the DA market clearing in \eqref{DA} determines the schedules of traditional generators and RES, cost recovery specifically ensures that their profits remain nonnegative given the market clearing schedules, 
\begin{equation}\label{G_profit}
    \sum_{\tau=1}^T(\lambda_{d,\tau,n}^*-\rho_n)p_{d,\tau,n}^* \geq 0
\end{equation}

\begin{equation}\label{W_profit}
    \sum_{\tau=1}^T \lambda_{d,\tau,n}^* w_{d,\tau,n}^* \geq 0
\end{equation}

To prove that conditions \eqref{G_profit} and \eqref{W_profit} hold, we derive the optimal dual objectives of problems \eqref{GP} and \eqref{WP}. We get,
\begin{equation}\label{G_profit_dual}
    \sum_{\tau=1}^T \overline{\eta}^*_{d,\tau,n}\overline{p}_n+\sum_{\tau=2}^T(\overline{\psi}_{d,\tau,n}^*+\underline{\psi}_{d,\tau,n}^*)\overline{r}_n
\end{equation}
\begin{equation}\label{W_profit_dual}
\sum_{\tau=1}^T\hat{y}_{d,\tau,n}\overline{\delta}^*_{d,\tau,n}
\end{equation}

Since the optimal dual solutions $\overline{\eta}^*_{d,\tau,n},\overline{\psi}_{d,\tau,n}^*,\underline{\psi}_{d,\tau,n}^*,\overline{\delta}^*_{d,\tau,n} \geq 0$, \eqref{G_profit_dual} and \eqref{W_profit_dual} are nonnegative. Also, since the strong duality holds for the convex problems of \eqref{GP} and \eqref{WP}, the optimal dual objectives \eqref{G_profit_dual},\eqref{W_profit_dual} equal the optimal primal objectives on the left of \eqref{G_profit} and \eqref{W_profit}, respectively. Therefore, the conditions  \eqref{G_profit} and \eqref{W_profit} hold.

\textbf{(b) Proof of revenue adequacy}

Revenue adequacy implies that, at the optimal solution, the total payment made by the load to the market operator equals the total payment made by the market operator to traditional and renewable generators , as well as to the transmission line operator. That is,
\begin{equation}\label{rev balance}
\begin{aligned}
&\sum_{n=1}^N \sum_{\tau=1}^T \lambda^*_{d,\tau,n} l_{d,\tau,n}  = \\
&\sum_{n=1}^N \sum_{\tau=1}^T [\lambda^*_{d,\tau,n}p^*_{d,\tau,n}+\lambda^*_{d,\tau,n}w^*_{d,\tau,n}+\sum_{m \in \Omega_n} \lambda^*_{d,\tau,n} f_{d,\tau,mn}^*]
\end{aligned}
\end{equation}
where $\Omega_n$ is the set of nodes connected to the node $n$ and $f_{d,\tau,mn}^*$ is the power flow from the node $m$ to the node $n$. In particular, the transmission line operator acts as a spatial arbitrageur, purchasing power at a lower-priced bus and selling it at a higher-priced bus.

The nodal power balance implies that,
\begin{equation}\label{Neq}
\begin{aligned}    &l_{d,\tau,n}=p^*_{d,\tau,n}+w^*_{d,\tau,n}+\sum_{m \in \Omega_n} f_{d,\tau,mn}^*,\\
&\forall \tau=1,...,T,n=1,...,N.
\end{aligned}
\end{equation}

By multiplying the nodal equalities \eqref{Neq} with the price $\lambda_{d,\tau,n}^*$ and summing them up across all nodes and time periods, we get \eqref{rev balance}.

\section{Derivation of the functions $f_{d,\tau}^x(\hat{\bm{y}}_d)$, $f_{d,\tau-1}^p(\hat{\bm{y}}_d)$, $f_{d,\tau}^z(\hat{\bm{y}}_d)$, and $f_{d,\tau-1}^{+-}(\hat{\bm{y}}_d)$}\label{policy_apendix}
In the following, we show how to derive the functions  $f_{d,\tau}^x(\hat{\bm{y}}_d)$, $f_{d,\tau-1}^p(\hat{\bm{y}}_d)$, $f_{d,\tau}^z(\hat{\bm{y}}_d)$, and $f_{d,\tau-1}^{+-}(\hat{\bm{y}}_d)$.
We define an operator $\Pi_{\mathcal{J}}: h(\bm{x}) \mapsto \Tilde{h}(\bm{x})$, where $h(\bm{x})=\bm{A}\bm{x}+\bm{b}$ is a linear function with the parameter $\bm{A},\bm{b}$, and $\mathcal{J}$ is the row index subset of $\bm{A},\bm{b}$. The output $\Tilde{h}(\bm{x})$ of the operator is also a linear function, where $\Tilde{h}(\bm{x})=\bm{A}[\mathcal{J}]\bm{x}+\bm{b}[\mathcal{J}]$ and $\bm{A}[\mathcal{J}],\bm{b}[\mathcal{J}]$ are the sub-matrix and sub-vector of $\bm{A},\bm{b}$.

Let $\mathcal{I}_{\text{DA},d,\tau}$ denote the row index set corresponding to  $\bm{x}_{d,\tau}^*$ within $\bm{x}_d^*$. The function that maps $\hat{\bm{y}}_d$ to $\bm{x}_{d,\tau}^*$ is,
\begin{equation}\label{functionxd}
 \begin{split}   &f^x_{d,\tau}(\hat{\bm{y}}_d):=\bm{x}_{d,\tau}^*=\Pi_{\mathcal{I}_{\text{DA},d,\tau}}(f^x_{d}(\hat{\bm{y}}_d)),\forall \tau=1,...,T
 \end{split}
\end{equation}
The coefficients are determined by the coefficients of $f^x_{d}(\hat{\bm{y}}_d)$ in \eqref{DApolicy}, whose row indexes belong to the set $\mathcal{I}_{\text{DA},d,\tau}$. By substituting \eqref{functionxd} into \eqref{RTpolicy}, we can obtain the function between RT primal solution $\bm{z}_{d,\tau}^*$ at time $\tau=1$ and the forecast $\hat{\bm{y}}_d$.
\begin{equation}\label{functionz}
\begin{split}
    &f_{d,\tau}^z(\hat{\bm{y}}_d):=\bm{z}_{d,\tau}^{*}=\bm{G}_{\text{RT}}[\mathcal{J}_{\text{RT},d,\tau}^a]^{-1}\bm{F}_{\text{RT}}[\mathcal{J}^a_{\text{RT},d,\tau}]f_{d,\tau}^x(\hat{\bm{y}}_d)+\\ &\bm{G}_{\text{RT}}[\mathcal{J}_{\text{RT},d,\tau}^a]^{-1}\bm{\psi}_{\text{RT},d,\tau}[\mathcal{J}_{\text{RT},d,\tau}^a], \tau=1
\end{split}
\end{equation}

To obtain the function between RT primal solution $\bm{z}_{d,\tau}^*$ at time-slot $\tau=2,...,T$ and the forecast $\hat{\bm{y}}_d$, we need to obtain the function between DA solution $\bm{p}_{d,\tau-1}^*$, RT solution $\bm{p}_{d,\tau-1}^{+-*}$ and the forecast $\hat{\bm{y}}_d$ as well. Concretely, since $\bm{p}_{d,\tau-1}^*$ is a part of $\bm{x}_{d,\tau-1}^*$, let $\mathcal{I}_{\text{DA},d,\tau-1}^p$ be the row index set corresponds to $\bm{p}_{d,\tau-1}^*$ within $\bm{x}_{d,\tau-1}^*$. With \eqref{functionxd}, the function between $\bm{p}_{d,\tau-1}^*$ and the forecast $\hat{\bm{y}}_d$ is,
\begin{equation}\label{functionp}
 \begin{split}   &f^p_{d,\tau-1}(\hat{\bm{y}}_d):=\bm{p}_{d,\tau-1}^*=\Pi_{\mathcal{I}_{\text{DA},d,\tau-1}^p}(f_{d,\tau-1}^x(\hat{\bm{y}}_d)),\\
 &\forall \tau=2,...,T
 \end{split}
\end{equation}

Likewise, since $\bm{p}_{d,\tau-1}^{+-*}$ is a part of $\bm{z}_{d,\tau-1}^*$, let $\mathcal{I}_{\text{RT},d,\tau-1}^{+-}$ be the row index set corresponds to  $\bm{p}_{d,\tau-1}^{+-*}$ within $\bm{z}_{d,\tau-1}^*$. We can express the function of $\bm{p}_{d,\tau-1}^{+-*}$ w.r.t. $\hat{\bm{y}}_d$ as,
\begin{equation}\label{functionadj}
\begin{split}
    &f_{d,\tau-1}^{+-}(\hat{\bm{y}}_d):=\bm{p}_{d,\tau-1}^{+-*}=\Pi_{\mathcal{I}_{\text{RT},d,\tau-1}^{+-}}(f_{d,\tau-1}^z(\hat{\bm{y}}_d)),\\
    &\forall \tau=2,...,T
\end{split}
\end{equation}

% By substituting \eqref{functionxd} $\forall \tau=1$ into \eqref{functionadj}, the function between $\bm{p}_{d,\tau-1}^{+-*},\forall \tau=2$ and $\hat{\bm{y}}_d$ is obtained, i.e., 
% \begin{equation}\label{functionadj2}
% \begin{split}
%     &f_{d,\tau-1}^{+-}(\hat{\bm{y}}_d):=\\
%     &\bm{p}_{d,\tau-1}^{+-*}=[\bm{G}_{\text{RT},\mathcal{J}_{\text{RT},d,\tau-1}^a}^{-1}\bm{F}_{\text{RT},\mathcal{J}^a_{\text{RT},d,\tau-1}}f_{d,\tau-1}^x(\hat{\bm{y}}_d)+\\ &\bm{G}_{\text{RT},\mathcal{J}_{\text{RT},d,\tau-1}^a}^{-1}\bm{\psi}_{\text{RT},\mathcal{J}_{\text{RT},d,\tau-1}^a}]_{\mathcal{I}_{\text{RT},d,\tau-1}^{+-}},\forall \tau=2
% \end{split}
% \end{equation}

Accordingly, by substituting  \eqref{functionxd}, \eqref{functionp}, \eqref{functionadj} into \eqref{RTpolicy2}, the function between $\bm{z}_{d,\tau}^*$ and forecast $\hat{\bm{y}}_d$ at $\tau=2,...,T$ is,
\begin{equation}\label{functionz2}
    \begin{split}
    &f^z_{d,\tau}(\hat{\bm{y}}_d):=\\
    &\bm{z}_{d,\tau}^*=\bm{G}_{\text{RT}}^{\prime}[\mathcal{J}_{\text{RT},d,\tau}^{ a}]^{-1}(\bm{\psi}^\prime_{\text{RT},d,\tau}[\mathcal{J}_{\text{RT},d,\tau}^{ a}]+\\
    &\bm{F}_{\text{RT}}^{\prime x}[\mathcal{J}_{\text{RT},d,\tau}^{ a}]f_{d,\tau}^x(\hat{\bm{y}}_d)+\bm{F}_{\text{RT}}^{\prime p}[\mathcal{J}_{\text{RT},d,\tau}^{ a}]f_{d,\tau-1}^p(\hat{\bm{y}}_d)+\\
    &\bm{F}_{\text{RT}}^{\prime +-}[\mathcal{J}_{\text{RT},d,\tau}^{ a}]f_{d,\tau-1}^{+-}(\hat{\bm{y}}_d)),\forall \tau=2,...,T
    \end{split} 
    \end{equation}

To sum up, the linear functions between the forecast $\hat{\bm{y}}_d$ and DA and RT primal solutions, which are defined in the neighborhood of $\hat{\bm{y}}_d$, are summarized in the \eqref{functionxd}-\eqref{functionz2}. 

\begin{figure}
  \centering
  \includegraphics[scale=0.55]{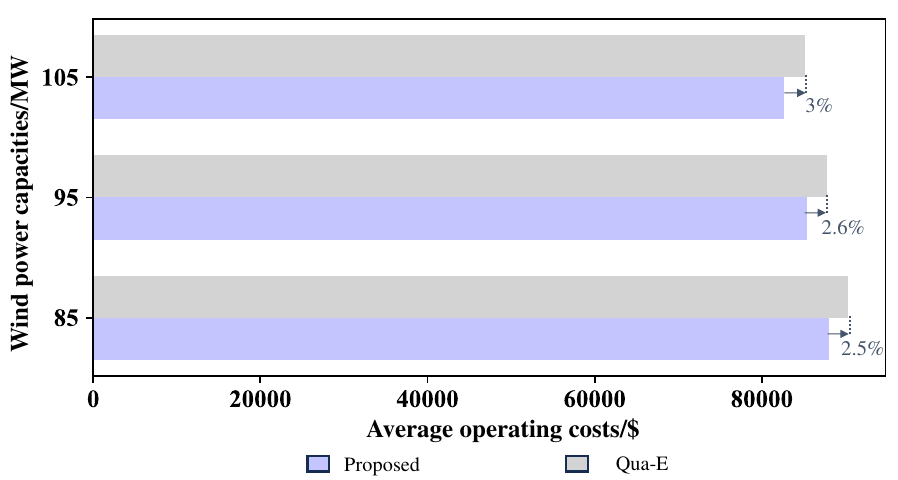}
  \caption{Average operating cost under different wind power capacities on the training set.}
\label{training}
\end{figure}

\begin{figure}
\centering
\includegraphics[scale=0.6]{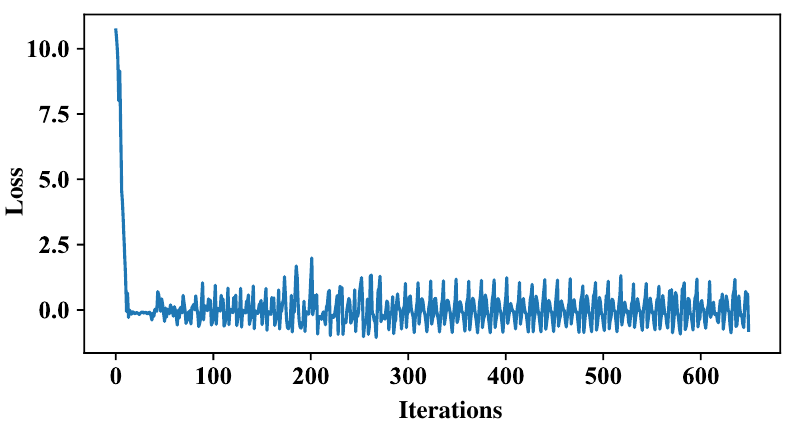}
\caption{The training convergence curve.}
\label{converge}
\end{figure}

\section{Results on the training set}\label{training results}

To demonstrate our model is not overfitting, we also report the results on the training set, which can be found in Fig. \ref{training}. The improvement achieved by the proposed method over Qua-E on the training set is similar to that observed on the test set in Fig. \ref{Capacities}. This demonstrates the model does not suffer overfitting. The results are reasonable since we train the forecasting model on a relatively large dataset with 7008 samples. Also, we use ResNet as the forecasting model, which avoids overfitting by using residual connections and allows the network to learn identity mappings and promote better gradient flow.

Additionally, to demonstrate the convergence of the forecasting model, we report the training convergence curve, which illustrates the change in the loss \eqref{valueloss} over training iterations. Since \eqref{valueloss} is a linear function composed of linear and constant parts, minimizing it is equivalent to minimizing the linear part. Therefore, the linear part itself can reflect the convergence process, which is shown in Fig. \ref{converge}. The loss decreases and converges rapidly over iterations.

\end{document}